\documentclass[a4paper,UKenglish,cleveref, autoref, thm-restate,authorcolumns]{lipics-v2019}

\usepackage{xspace}
\usepackage{comment}
\usepackage{amsmath, amssymb, amsfonts}
\usepackage{enumerate}
\usepackage{verbatim}
\usepackage{hyperref}
\usepackage{todonotes}
\usepackage[utf8]{inputenc}

\newcommand{\todom}[1]{\todo[color=green]{#1}}
\newcommand{\todoa}[1]{\todo[color=magenta]{#1}}

\newcommand{\prob}{\mathbb{P}}

\title{Entropy bounds for grammar compression} 

\author{Micha\l{} Ga\'nczorz}
{Institute of Computer Science, University of Wroc\l{}aw, Poland }{mga@cs.uni.wroc.pl}{https://orcid.org/0000-0003-4047-2486}{}
\authorrunning{M. Ga\'nczorz}
\Copyright{Micha\l{} Ga\'nczorz}



\category{} 

\relatedversion{} 

\supplement{}

\funding{Supported under National Science Centre, Poland project number 2017/26/E/ST6/00191.}

\EventEditors{John Q. Open and Joan R. Access}
\EventNoEds{2}
\EventLongTitle{42nd Conference on Very Important Topics (CVIT 2016)}
\EventShortTitle{CVIT 2016}
\EventAcronym{CVIT}
\EventYear{2016}
\EventDate{December 24--27, 2016}
\EventLocation{Little Whinging, United Kingdom}
\EventLogo{}
\SeriesVolume{42}
\ArticleNo{23}

\hideLIPIcs
\nolinenumbers

\newcommand{\dBref}[1]{(\hyperref[dB#1]{dB#1})\xspace}

\newcommand{\IGref}[1]{(\hyperref[IG#1]{IG#1})\xspace}
\newcommand{\IGreftwo}[2]{(\hyperref[IG#1]{IG#1--IG#2})\xspace}
\newcommand{\IGrefall}{(\hyperref[IG1]{IG1--IG3})\xspace}

\usepackage{color}
\newcommand{\usecomment}[1]{#1}
\definecolor{darkgreen}{rgb}{0.0,0.7,0.0}
\newenvironment{aj}{\noindent\color{magenta} Artur:} {}
\newcommand{\artur}[1]{\usecomment{\begin{aj} #1\end{aj}\xspace}}
\newenvironment{mg}{\noindent\color{blue} Michał:} {}
\newcommand{\michal}[1]{\usecomment{\begin{mg} #1\end{mg}\xspace}}

\newcommand{\nont}{{\mathcal N}}
\newcommand{\nonterminals}[1]{{|\nont(#1)|}}
\DeclareMathOperator{\grhs}{{rhs}}
\newcommand{\rhs}[1]{|\grhs(#1)|}

\newcommand{\seq}[3]{\ensuremath{{#1}_{#2},\ldots,{#1}_{#3}}}

\newcommand{\algofont}[1]{\textnormal{\selectfont\sffamily#1}}
\providecommand{\twodots}{\mathinner{\ldotp\ldotp}}

\providecommand{\Ocomp}{{\mathcal O}}

\DeclareMathOperator{\gc}{gc}

\newcommand{\RePair}{\algofont{Re-Pair}\xspace}
\newcommand{\Greedy}{\algofont{Greedy}\xspace}
\newcommand{\Sequitur}{\algofont{Sequitur}\xspace}
\newcommand{\Sequential}{\algofont{Sequential}\xspace}
\newcommand{\LongestMatch}{\algofont{LongestMatch}\xspace}
\newcommand{\lzss}{\algofont{LZ77}\xspace}
\newcommand{\lzso}{\algofont{LZ78}\xspace}

\newcommand{\cnfenc}{{incremental}\xspace}

\DeclareMathOperator{\Lengths}{Lengths}

\newlength{\figurewidth}
\newlength{\smallfigurewidth}

\begin{document}


\maketitle

\begin{abstract}
Grammar compression represents a string as a context-free grammar.
Achieving compression requires encoding such grammar
as a binary string; there are a few commonly used encodings.
We bound the size of practically used encodings for several heuristical compression methods,
including \RePair and \Greedy algorithms:
the standard encoding of \RePair,
which combines entropy coding and special encoding of a grammar,
achieves $1.5|S|H_k(S)$, where $H_k(S)$ is $k$-th order entropy of $S$.
We also show that by stopping after some iteration we can achieve $|S|H_k(S)$.
This is particularly interesting, as it explains
a phenomenon observed in practice:
introducing too many nonterminals causes the bit-size to grow.
We generalize our approach to
other compression methods
like \Greedy and a wide class of irreducible grammars
as well as to other practically used bit encodings
(including naive, which uses fixed-length codes).
Our approach not only proves the bounds
but also partially explains why \Greedy and \RePair
are much better in practice than other grammar based methods.
In some cases, we argue that our estimations are optimal.
The tools used in our analysis are of independent interest:
we prove the new, optimal bounds
on the zeroth-order entropy of the parsing of a string.
\end{abstract}


\section{Introduction}
Grammar compression is a type of dictionary compression
in which we represent the input as a context-free grammar
producing exactly the input string.
Computing the smallest grammar
(according to the sum of productions lengths)
is known to be NP-complete~\cite{SmallestGrammar,SLPapproxNPhard,SLPNPnew}
and best approximation algorithms~\cite{SLPaprox,SmallestGrammar,simplegrammar,SLPaproxSakamoto}
have logarithmic approximation ratios,
which renders them impractical.
On the other hand, we have heuristical algorithms~\cite{RePair,Sequitur,apostolico1998some,kieffer2000grammar,lz78,RePairforEntropy},
which achieve compression ratios comparative to other dictionary methods~\cite{RePair,apostolico1998some},
despite having poor theoretical guarantees~\cite{SmallestGrammar,SLPaproxrevisited}.
Heuristics use various bit encodings,
e.g.\ 
Huffman coding
~\cite{RePair,apostolico1998some},
which is standard for other compression methods as well.

The apparent success of heuristics fuelled theoretical
research that tried to explain or evaluate their performance:
one branch of such an approach tried to estimate 
their approximation ratios
~\cite{SmallestGrammar,SLPaproxrevisited}
(no heuristic was shown to achieve logarithmic ratio though),
the other tried to estimate the bit-size of their output,
which makes it possible to compare them with other compressors as well.

In general, estimating the bit-size of any compressor is hard.
However, this is not the case for compressors based on the (higher-order) entropy,
which includes Huffman coding, arithmetical coding, \algofont{PPM}, and others.
In their case, the bit-size of the output is very close to the ($k$-order) empirical entropy
of the input string ($H_k$),
thus instead of comparing to those compressors, we can compare the size of
the output with the $k$-order entropy of the input string. 
Moreover, in practice, it seems that $k$-order entropy is a good estimation
of possible compression rate of data
and entropy-based compression is widely used.
Furthermore, there are multiple results that tie the output of
a compressor to $k$-order entropy.
Such analysis was carried out for
\algofont{BWT}~\cite{ManziniBurrowsWheeler}, \lzso~\cite{KosarajuManzini99}, \lzss~\cite{szpankowski1993asymptotic} compression methods.
In some sense,
the last two
generalised classical results
from information theory on \algofont{LZ}-algorithms coding for ergodic sources~\cite{wyner1994sliding,lz78},
a similar work was also performed for a large class of grammar compressors~\cite{kieffer2000grammar} with a certain encoding.


When discussing the properties of grammar compressors,
it is crucial to distinguish the grammar construction and grammar bit-encoding.
There are much fewer results that link the performance of grammar compressors
to $k$-th order entropy than to optimal grammar size,
despite the wide popularity of grammar-based methods
and the necessity of using proper bit encodings in any practical compressor.
The possible reasons are that heuristical compressors often employ their own,
practically tested and efficient, yet difficult to analyze, encodings.

\paragraph{Our contribution}
The main result of this work is showing that
a large spectrum of grammar compressors combined
with different practically-used encodings achieve bounds related to $|S|H_k(S)$.
The considered encodings are:
fully naive, in which each letter and nonterminal 
is encoded with code of length $\lceil \log |G| \rceil$;
naive, in which the starting production string is encoded with an entropy coder
and remaining grammar rules are encoded as in the fully naive variant;
incremental, in which the starting string is encoded with an entropy coder
and the grammar is encoded separately with a specially designed encoding
(this is the encoding used originally for \RePair~\cite{RePair});
and entropy coding, which entropy-encodes the concatenation
of productions' right-hand sides.

First, we show that 
\RePair stopped at the right moment achieves
$|S|H_k(S) + o(|S| \log \sigma)$ bits for $k = o(\log_\sigma |S|)$
for naive, incremental, and entropy encoding.
Moreover, at this moment, the size of the dictionary is $\Ocomp(|S|^c)$, $c < 1$,
where $c$ affects linearly the actual value
hidden in $o(|S| \log \sigma)$.
This implies that grammars produced by \RePair
and similar methods (usually) have a small alphabet,
while many other compression algorithms,
like~\lzso, do not have this property~\cite{lz78};
this is of practical importance, as encoding a~dictionary can be costly.
Then we prove that in general \RePair's output size can be bounded by $1.5|S|H_k(S) +
o(|S| \log \sigma)$ bits for both incremental and entropy encoding.

For \Greedy we give similar bounds:
it achieves $1.5|S|H_k(S) + o(|S|\log \sigma)$
bits with entropy coding,
$2|S|H_k(S) + o(|S|\log \sigma)$ using fully naive encoding
and if stopped after $\Ocomp(|S|^c)$ iterations
it achieves $|S|H_k(S) + o(|S|\log \sigma)$ bits with entropy encoding.
The last result is of practical importance:
each iteration of \Greedy requires $\Ocomp(|S|)$ time,
and so we can reduce the running time from $\Ocomp(|S|^2)$ to $\Ocomp(|S|^{1+c})$
while obtaining comparable, if not better, compression rate.

Stopping the compressor during its runtime seems counter-intuitive 
but is consistent with empirical observations~\cite{DBLP:conf/cpm/GonzalezN07,ImprovementsRePairGrammar}
and our results shed some light on the reasons behind this phenomenon.
Furthermore, there are approaches suggesting partial decompression
of grammar-compressors in order to achieve better (bit) compression rate~\cite{ruleReductionSequential}.

Yet, as we show, this does not generalize to all grammar compressors:
in general, stopping a grammar compressor after a certain iteration
or partially decompressing the grammar such that the new grammar has
at most  $\Ocomp(|S|^c)$ nonterminals, $c < 1$,
could yield a grammar of size $\Theta(|S|)$.
In particular, we show that there are \emph{irreducible grammars}~\cite{kieffer2000grammar}
which have this property.
Comparing to \RePair and \Greedy, both while stopped after $\Ocomp(|S|^c)$ iterations
produce a \emph{small} grammar of size at most $\Ocomp(|S|/\log_\sigma |S|)$
which is an information theoretic bound for a grammar size for any $S$.
Moreover, the fact that the grammar is small is
crucial in proving the aforementioned bounds.
This hints why \RePair and \Greedy in practice produce grammars of smaller size
compared to other grammar compressors:
they significantly decrease the size of the grammar in a small number of steps.

Finally, we apply our methods to the general class of irreducible grammars~\cite{kieffer2000grammar}
and show that entropy coding of an irreducible grammar 
uses at most $2 |S|H_k(S) + o(|S|\log \sigma )$ bits
and fully naive encoding uses 
at most $6 |S|H_k(S) + o(|S|\log \sigma)$ bits.
No such general bounds were known before.

To prove our result we show a generalization of
the result by Kosaraju and Manzini~\cite[Lem.~2.3]{KosarajuManzini99},
which in turn generalizes Ziv's Lemma~\cite[Section~13.5.5]{ElementfofInformationTheory2006}.
\begin{theorem}[cf.~{\cite[Lem.~2.3]{KosarajuManzini99}}]
	\label{thm:main_estimation}
	Let $S$ be a string over $\sigma$-sized alphabet,
	$Y_S = y_1 y_2 \ldots y_{|Y_S|}$ its parsing,
	and $L = \Lengths(Y_S)$ is a word $|y_1| |y_2| \ldots  |y_{|Y_S|}|$
	over the alphabet $\{1, \twodots, |S|\}$.~Then: 
	\begin{equation*}
	|Y_S| H_0(Y_S)
	\leq 
	|S|H_k(S) + |Y_S|k\log \sigma + |L| H_0(L) \enspace .
	\end{equation*}
	Moreover, if $k = o(\log_\sigma |S|)$
	and $|Y_S| = \Ocomp\left({|S|}/{\log_\sigma |S|}\right)$ then:
	\begin{equation*}
	|Y_S| H_0(Y_S)
	\leq 
	|S|H_k(S) + \Ocomp\left(
	\frac{|S| k \log \sigma}{\log_\sigma |S|}
	+ \frac{|S|\log \log_\sigma |S|}{\log_\sigma |S|}
	\right)
	\leq
	|S|H_k(S) + o(|S| \log \sigma) \enspace ;
	\end{equation*}
	the same bound applies to the Huffman coding of string $S'$ obtained by replacing
	different phrases of $Y_S$ with different letters. 
\end{theorem}

Comparing to~\cite[Lem.~2.3]{KosarajuManzini99},
our proof is more intuitive, simpler
and removes the greatest pitfall of the previous result:
the dependency on the highest frequency of a phrase in the parsing.

The idea of a parsing is strictly related to dictionary compression methods,
as most of these methods pick some substring and replace it with a new symbol,
creating a phrase in a parsing (e.g.\ Lempel-Ziv algorithms).
Thus, we use Theorem~\ref{thm:main_estimation} to estimate the size of the entropy-coded output
of grammar compressors.
We also use Theorem~\ref{thm:main_estimation} in a less trivial way:
with its help, we can \emph{lower bound} the zeroth-order entropy of the string,
we use such lower bounds when proving upper-bounds for both \Greedy and \RePair.
This is one of the crucial steps in showing $1.5|S|H_k(S) + o(|S|\log\sigma)$ bound
as it allows us to connect the entropy of the string with bit encoding of the grammar.
The general idea is that when estimating the lower bound of string $S$,
it is sometimes easier to lower bound the entropy of some parsing of $S$,
thus we set $k=0$ in Theorem~\ref{thm:main_estimation}
and show that there exists a parsing with large entropy.
By Theorem~\ref{thm:main_estimation} it follows that zeroth-order entropy
of a parsing should lower bound 
the entropy of the string (modulo small factors).

An interesting corollary of Theorem~\ref{thm:main_estimation}
is that, for $k = 0$, for \emph{any} parsing $Y_S$ the entropy $|Y_S|H_0(Y_S)$
may be larger than $|S|H_0(S)$ by at most $|L| H_0(L)$,
which is $o(S)$ for $|L|=o(|S|)$,
i.e.\ applying parsing to $S$ cannot increase encoding cost significantly.

Theorem~\ref{thm:main_estimation} 
is an important tool to connect parsing entropy to $k$-order entropy of the input,
as it holds for \emph{every} parsing.
On the other hand, it is not difficult to give examples of strings $S$ and their parsings $Y_S$
for which the bound from Theorem~\ref{thm:main_estimation} is not tight,
i.e.\ there are strings where the~inequality holds even without the $|Y_S|k\log \sigma$ summand.
One such example is de Bruijn sequence with any parsing where phrases
are no longer than $\log_\sigma |S|$.
Thus it is natural to ask whether the bound can be improved,
e.g.\ if we are allowed to \emph{choose} a parsing.
Theorem~\ref{thm:mean_entropy} gives a (partial) positive answer.

\begin{theorem}
	\label{thm:mean_entropy}
	Let $S$ be a string over $\sigma$-sized alphabet.
	Then for any integer $l$
	we can construct a~parsing $Y_S = y_1 y_2 \ldots y_{|Y_S|}$ of size $|Y_S| \leq
	\left\lceil |S|/l \right\rceil + 1$ satisfying:
	\begin{equation*}
	|Y_S| H_0(Y_S)
	\leq
	|S| \frac{\sum_{i=0}^{l-1} H_i(S)}{l} + \Ocomp(\log |S|)
	\enspace . \end{equation*}
	Moreover, $|y_1|, |y_{|Y_S|}|\leq l$ and all the other phrases have length $l$.
\end{theorem}
Theorem~\ref{thm:mean_entropy}
can be readily applied to text encodings with random access based on parsing
the string into equal-length phrases~\cite{FerraginaV07SimpStat,GonzalezNStatistical,GrossiDynamicIndexes}
and improve their performance guarantee
from
$|S|H_k(S) + \Ocomp\left(|S|/\log_\sigma |S| \cdot  (k \log \sigma + \log \log |S|)  \right)$
to $\frac{|S|}{l} \sum_{i=0}^{l-1} H_i(S) + \Ocomp(\log |S| + |S|\log \log |S|/\log_\sigma |S|)$
(such structures consists of entropy coded string and additional structures
which use $\Ocomp\left(|S| \log \log |S|/\log_\sigma |S| \right)$ bits).

As implied by the recent work~\cite{ganczorzEntropyLowerBounds},
the bounds neither in Theorem~\ref{thm:main_estimation} nor in Theorem~\ref{thm:mean_entropy}
can be improved,
moreover some of our estimations
on grammars' size are optimal as well
(e.g.\ those achieving $|S|H_k(S) + o(|S|\log\sigma)$).

\paragraph{Related results}
With the use of previously known tools~\cite{KosarajuManzini99}
it was shown that
\RePair with fully naive encoding achieves
$2|S|H_k(S) + o(|S| \log \sigma)$ bits for $k = o(\log_\sigma n)$~\cite{RePairentropy}.
Recently it was shown~\cite{ochoa2018irreducible}
that irreducible grammars with the encoding of Kieffer and Yang~\cite{kieffer2000grammar}
achieve $|S|H_k(S) + o(|S| \log \sigma)$
this was obtained by adapting
the aforementioned results of Kieffer and Yang~\cite{kieffer2000grammar}.
Note that this result is a simple corollary of Theorem~\ref{thm:main_estimation},
which predates that result:
for a string $S$ Kieffer and Yang showed~\cite{kieffer2000grammar} that their encoding
consumes at most $|Y_S|H_0(Y_S) + \Ocomp(|Y_S|)$ bits
for a parsing $Y_S$ induced by an irreducible grammar;
they also showed that $|Y_S| = \Ocomp (|S|/\log_\sigma |S|)$.
Combining those facts with Theorem~\ref{thm:main_estimation} yields the bound
of $|S|H_k(S) + o(|S|\log \sigma)$.

While the bound from~\cite{ochoa2018irreducible} seems stronger at a first sight,
it is of a different flavour:
it applies to a sophisticated encoding that is not used in practice
and it requires a grammar to be in the irreducible form,
so, for instance, it applies to \RePair only after some postprocessing of the grammar
(which can completely change grammar structure)
and for \Greedy it is assumed that the algorithm is run to the end,
while the original formulation suggested stopping the compressors after 
certain iteration~\cite{apostolico1998some,apostolicoGreedy2}.
Furthermore, it seems hard to apply Kieffer and Yang encoding to compressed data structures.
Lastly, the bound seems to be weakly connected to real-life performance,
as it qualifies a wide range of algorithms into the same class 
despite the fact that many of them 
perform differently in practice.

The last disadvantage seems to be especially important:
by exploiting the properties of \RePair and \Greedy
we were able to prove that they 
achieve $|S|H_k(S)$ with the factor $1.5$ compared
to $2$ for the irreducible grammars with the entropy encoding
(and for $k=0$ it can be shown that these bounds are tight);
moreover, we show that \RePair and \Greedy stopped after certain iteration have a small dictionary, again, contrary to irreducible grammars.

\section{Strings and their parsing}
\label{sec:definitions}
%
%
A string is a sequence of elements, called \emph{letters},
from a finite set, called \emph{alphabet},
and it is denoted as $w = w_1w_2\cdots w_k$,
where each $w_i$ is a letter, its length $|w|$ is $k$;
alphabet's size is denoted by $\sigma$,
the alphabet is usually clear from the context, 
$\Gamma$ is used when some explicit name is needed.
We often analyse words over different alphabets.
For two words $w,w'$ the $ww'$ denotes their concatenation.
By $w[i \twodots j]$ we denote $w_iw_{i+1}\cdots w_{j}$,
this is a \emph{subword} of $w$;
$\epsilon$ denotes the empty word.
For a pair of words  $v, w$ $|w|_v$ denotes
the number of different (possibly overlapping) subwords of $w$ equal to $v$;
if $v = \epsilon$ then for uniformity of presentation
we set $|w|_\epsilon = |w|$. 
Usually, $S$ denotes the input string.

A grammar compression represents an input string $S$
as a context-free grammar generating a unique string $S$.
The right-hand side of the start nonterminal is called
a \emph{starting string} (often denoted as $S'$)
and by the grammar $G$ we mean the collection of other rules,
together they are called the \emph{full grammar} and denoted as $(S', G)$.
For a nonterminal $X$ we denote its rule right-hand side by
$\grhs(X)$.
For a grammar $G$ its right-hand sides, denoted by $\grhs(G)$,
is the set of $G$'s productions' right-hand sides,
for a full grammar $(S', G)$ we also add $S'$ to the set
and denote it by $\grhs(S',G)$.
We denote by $\rhs{G}$ ($\rhs{S',G}$) the sum of strings' lengths in $\grhs(G)$ ($\grhs(S',G)$, respectively).
By $\nonterminals{G}$ we denote the number of nonterminals.
A grammar is in CNF if all right-hand sides of a~grammar consist of two symbols. 
The string generated by a nonterminal $X$ in a grammar $G$
is the \emph{expansion} $\exp(X)$ of $X$ in $G$.
All reasonable grammar compressions guarantee that
no two nonterminals have the same expansion, $\rhs{X} > 1$ for each nonterminal $X$ 
and that each nonterminal occurs in the derivation tree;
we implicitly assume this in the following.
A grammar (full grammar) $G$ ($(S', G)$, respectively) is said to be \emph{small},
if $\rhs{G} = \mathcal{O}({|S|}/{\log_\sigma |S|})$
($\rhs{S', G} = \mathcal{O}({|S|}/{\log_\sigma |S|})$, respectively).
This matches the folklore information-theoretic lower bound on the size of a grammar for a string of length $|S|$.

In practice, the starting string and the grammar may be encoded in different ways,
especially when $G$ is in CNF, hence we make a distinction between these two.
Note that for many (though not all) grammar compressors
both theoretical considerations and proofs as well as practical evaluation
show that the size of the grammar is considerably smaller than the
size of the starting string.

A \emph{parsing} of a string $S$ is any representation
$S = y_1y_2\cdots y_c$, where each $y_i \neq \epsilon$.
We call $y_i$ a phrase.
We denote a parsing as $Y_S = \seq y 1 c$
and treat it as a word of length $c$ over the alphabet $\{\seq y 1 c\}$;
in particular $|Y_S| = c$ is its size.
Then $\Lengths(Y_s) = |y_1|,|y_2|,\ldots, |y_c| \in \mathbb N^*$
and we treat it as a word over the alphabet $\{1,2,\ldots, |S|\}$.
In grammar compression,
the starting string $S'$ of a full grammar $(S', G)$ induces a parsing
of input string, also a concatenation of strings in $\grhs{(S',G)}$ induces
a parsing plus some additional nonterminals.

\begin{definition}
\label{def:k_order_entropy}
For a word $S=s_1 s_2 \ldots s_{|S|}$ its \emph{$k$-order empirical entropy} is:
\begin{equation*}
H_k(S)
= - \frac{1}{|S|} \sum_{i=k+1}^{|S|} \log \prob(s_i | S[i-k\ldots i-1])
\enspace,
\end{equation*}
where $\prob(s_i | w)$ is the empirical probability of letter $s_i$ occurring
in context $w$,
i.e.\ $\prob(s_i | S[i-k\ldots i]) = \frac{|S|_{S[i-k\ldots i]}}{|S|_{S[i-k\ldots i-1]}}$
if $S[i-k\ldots i-1]$ is not a suffix of $S$ or
$S[i-k\ldots i-1] = \epsilon$;
and $\prob(s_i | S[i-k\ldots i]) = \frac{|S|_{S[i-k\ldots i]}}{|S|_{S[i-k\ldots i-1]}-1}$
otherwise.
\end{definition}

We are mostly interested in the $H_k$ of the input string $S$
and in the $H_0$ for parsing $Y_S$ of $S$.
The first is a natural measure of the input,
to which we compare the size of the output of the grammar compressor,
and the second corresponds to the size of the entropy coding of the starting string returned by a grammar compressor.

\section{Entropy upper bounds on grammar compression}\label{sec:upper-bounds-on-grammar-compression}
In this Section, we use Theorem~\ref{thm:main_estimation}
to bound the size of grammars returned by popular compressors:
\RePair~\cite{RePair}, \Greedy~\cite{apostolico1998some}
and a general class of methods that produce irreducible grammars.
We consider a couple of natural and simple practically used bit-encodings of grammars.
Interestingly, we obtain bounds of~form $\alpha |S| H_k(S) + o(|S| \log \sigma)$
for a constant $\alpha$
for all of them.
\subsection{Encoding of grammars}
\label{subsec:encoding_of_grammars}
We first discuss the possible encodings of the (full) grammar
and give some estimations on their sizes.
All considered encodings assume a linear order on nonterminals:
if $\grhs(X)$ contains $Y$ then $X \geq Y$.
In this way, we can encode the rule
as a sequence of nonterminals on its right-hand side,
in particular instead of storing the nonterminal names
we store the positions in the above ordering.

The considered encodings are simple, natural
and correspond closely or exactly to encodings used
in grammars compressors
like \RePair~\cite{RePair} or \Sequitur~\cite{nevill1997compression}.
Some other algorithms, e.g.~\Greedy~\cite{apostolico1998some},
use specialized encodings,
but at some point, they still encode grammar 
using entropy coder with some additional
information or assign codes to each nonterminal in the grammar.
Thus, most of these custom encodings are roughly
equivalent to (or not better than) entropy~coding. 

Encoding of CNF grammars deserves special attention
and is a problem investigated on its own.
It is known that a grammar $G$ in CNF can be encoded using
$\nonterminals{G} \log (\nonterminals{G}) + 2\nonterminals{G} + o(\nonterminals{G})$ bits%
~\cite{tabei2013succinctGrammar,takabatake2012variable},
which is close to the information theoretic lower bound of 
$\log (\nonterminals{G}!) + 2\nonterminals{G}$~\cite{tabei2013succinctGrammar}.
On the other hand, heuristics use simpler encodings,
e.g.~\RePair was implemented and tested with several encodings,
which were based on a division of nonterminals of $G$ into $z$ groups
$g_1, \ldots, g_{z}$ where $X \in g_i$, $X\rightarrow AB$, if and only if
$A \in g_{i-1}$ and $B \in g_{j}$ or $B \in g_{i-1}$ and $A \in g_j$,
for some $j \leq i-1$, where $g_0$ is the input alphabet.
Then each group is encoded separately.
Though no theoretic bounds were given,
we show (in the Appendix) that some of
these encodings achieve previously mentioned lower bound of
$\log (\nonterminals{G}!) + 2\nonterminals{G}$ bits (plus smaller order terms).
The above encodings are difficult to analyse due to heuristical optimisations;
instead, we analyse an incremental encoding,
which is a simplified version of one of the original methods used to encode~\RePair output~\cite{RePair}.
It matches the theoretical lower bound except for a larger constant hidden in $\Ocomp(\nonterminals{G})$.
To be precise, we consider the following encodings:
\begin{itemize}
	\item \textit{Fully naive} 
	We concatenate the right-hand sides of the full
	grammar.
	Then each nonterminal and letter are assigned bitcodes of the same length. 
	We store $\rhs{X}$ for each $X$,
	as it is often small, it is sufficient to store it in~unary.
	\item \textit{Naive} The starting string is entropy-coded,
	the rules are coded as in the fully-naive variant.
	\item \textit{Entropy-coded} The rules' right-hand sides are concatenated with the starting string and they are coded using an entropy coder.
	\item \textit{Incremental} We use this encoding only for CNF grammars,
	though it can be extended to general grammars
	as it is possible to transform any grammar to such a form.
	Given a full grammar $(S', G)$ 
	we use entropy coder for $S'$ and encoding defined below for $G$.
	It has additional requirements on the order on letters and nonterminals:
	if 
	$X \leq Y$ and $X', Y'$ are the first symbols in productions
	for $X, Y$ then $X' \leq Y'$.
	Given any grammar, this property can be achieved by permuting
	the nonterminals,
	but we must drop the assumption that the right-hand side of a given nonterminal $X$
	occurs before $X$ in the sequence.
	Then the grammar $G$ can be viewed as a sequence of nonterminals:
	$X_1 \rightarrow Z_1 U_1, 
	\ldots , X_{\nonterminals{G}} \rightarrow Z_{\nonterminals{G}} U_{\nonterminals{G}}$.
	We encode differences $\Delta_{X_i} = Z_i - Z_{i-1}$ with Elias $\delta$-codes,
	and $U_i$'s naively using $\lceil \log (\nonterminals{G} + \sigma) \rceil$ bits.
\end{itemize}

We upper-bound the sizes of grammars under various encodings;
the first two estimations are straightforward,
the third one requires some calculations.
The estimation in Lemma~\ref{lem:entropy_coding_concatenation} requires a nontrivial argument.

\begin{lemma}
	\label{lem:grammar_encodings}
	Let $S$ be a string and $(S', G)$ a full grammar that generates it.
	Then:\\
	\emph{fully naive} uses at most $ \left( \rhs{S',G}\right) \left( \log(\sigma + \nonterminals{G}) + \Ocomp(\rhs{S',G}) \right)$ bits;\\
	\emph{naive} uses at most $|S'|H_0(S') + \rhs{G}\log(\sigma + \nonterminals{G}) + \Ocomp(\rhs{S',G})$ bits;\\
	\emph{\cnfenc} uses at most $|S'|H_0(S') + \nonterminals{G}  \log (\sigma + \nonterminals{G}) +
	\Ocomp(\rhs{S',G})$~bits.
\end{lemma}
The proof idea of Lemma~\ref{lem:entropy_coding_concatenation}
is to show that $\grhs{(S', G)}$ induces a parsing of $S$,
with at most $\nonterminals{G}$ additional symbols,
and then apply Theorem~\ref{thm:main_estimation}.
The latter requires that different nonterminals have different expansions,
all practical grammar compressors have this property.
\begin{lemma}
	\label{lem:entropy_coding_concatenation}
Let $S$ be a string over an alphabet of size $\sigma$,
$k = o(\log_\sigma |S|)$
and $(S', G)$ a~full grammar generating it,
where no two nonterminals have the same expansion.
Let $S_G$ be a~concatenation of strings from $\grhs(S',G)$.
If $|S_G| = \Ocomp\left( |S|/ \log_\sigma |S| \right)$
then
$
	|S_G| H_0(S_G) \leq |S|H_k(S) + \nonterminals{G}\log |S| +  o(|S| \log \sigma) 
$
, same bound applies to the entropy coding of~$(S',G)$.
\end{lemma}

\subsection{\RePair}\label{subsec:repair}
\RePair is one of the most known grammar compression heuristics.
It starts with the input string $S$
and in each step replaces a most frequent pair $AB$ in $S$ with a new symbol $X$
and adds a rule $X\rightarrow AB$.
For the special case when $A=B$
the replacements are realized by left-to-right scan,
i.e.\ maximal substring $A^n$ is replaced
by $X^{n/2}$ for even $n$ and $X^{(n-1)/2}A$ for odd $n$;
moreover, it never replaces the pair with only two overlapping occurrences.
The choice of a pair to replace is dependent only on the overall number of occurrences,
both overlapping and non-overlapping,
thus the algorithm may replace the pair $A'A'$ which has less
non-overlapping occurrences than some other pair $AB$.

\RePair is simple, fast, and provides
compression ratio better than some of the standard
dictionary methods like \algofont{gzip}~\cite{RePair}.
It found usage in various applications~\cite{DBLP:conf/cpm/GonzalezN07,Collagesystem,DBLP:conf/spire/ClaudeN07,DBLP:journals/tcs/FischerMN09,RaymondWanPhD,Gonzalez:2015:LCS:2627368.2594408}.

We refer to the current state of $S$,
i.e.\ $S$ with some pairs replaced, as the \emph{working string}.
We prove that \RePair stopped at the right moment
achieves $H_k$ (plus some smaller terms),
using any of the:
naive, \cnfenc
or entropy encoding.
To this end, we show that 
\RePair reduces the input string to length $\frac{\alpha |S|}{\log_\sigma |S|}$
for an appropriate $\alpha$
and stop the algorithm when the string gets below this size.
Then use the fact that different nonterminals produced by \RePair
have different expansions (\cite[Lem.~4]{RePairentropy}):
At this iteration, by estimations on the number of possible different substrings, 
we show that for an adequate
$\alpha, \beta$
there exist a pair of nonterminals $AB$, with $|\exp(A)|, |\exp(B)| < (\log_\sigma |S|)/\beta$,
such that pair $AB$ occurs frequently, i.e.\ $\approx|S|^{1-2/\beta}$.
Then, we use the fact that the frequency of the most frequent pair never increases
during the algorithm runtime ({\cite[Lem.~3]{RePairentropy}}),
and so grammar size at this point must be $\approx |S|^{2/\beta}$.
Then, on the one hand, we can show 
that the encoding of grammar constructed so far
takes at most $\approx |S|^{c} \log |S|  = o(|S|), \ c < 1$,
and on the other hand Theorem~\ref{thm:main_estimation}
yields that entropy coding of the working string
is $|S|H_k(S)$ plus some smaller terms.


\begin{theorem}
	\label{thm:RePairStopped}
	Let $S$ be a string over $\sigma$-size alphabet and $k = o(\log_\sigma |S|)$.
	Let $\xi > 0, \beta > 2$ be constants.
	When the size of the working string of \RePair is first
	below $\frac{(\xi + 3\beta) |S|}{\log_\sigma |S|} + 2|S|^\frac{2}{\beta} + 1$ then 
	the number of nonterminals in the grammar is at most
	$1 + \frac{2}{\xi}(|S|^\frac{2}{\beta}\cdot \log_\sigma |S|)$;
	for every $\xi>0, \beta>2$ and $S$ such a point always exists.
	Moreover, at this moment, the size of the entropy coding of the working string
	is at most $|S|H_k(S) + o(|S| \log \sigma)$
	and the bit-size of \RePair stopped at this point
	is at most $|S|H_k(S) + o(|S| \log \sigma)$
	for naive, entropy, and incremental encoding.
\end{theorem}

Theorem~\ref{thm:RePairStopped} says that \RePair
achieves $H_k$ when stopped at the appropriate time.
Surprisingly, continuing to the end can lead to a worse compression rate.
In fact, limiting the size of the dictionary for \RePair
(as well as for similar methods) in practice results in a better compression rate for larger
files~\cite{DBLP:conf/cpm/GonzalezN07,ImprovementsRePairGrammar}.
This is not obvious, in particular,
Larsson and Moffat~\cite{RePair} believed that it is the other way around;
this belief was supported by the results on smaller-size data.
This is partially explained by Theorem~\ref{thm:RePairToTheEnd},
in which we give a $1.5|S|H_k(S) + o(|S| \log \sigma)$ bound
on \RePair run to the end with \cnfenc or entropy encoding
(a $2|S|H_k(S) + o(|S| \log \sigma)$ bound for fully naive encoding was known earlier~\cite{RePairentropy}).

Before proving Theorem~\ref{thm:RePairToTheEnd}
we first show an example, see Lemma~\ref{lem:15entropyRePairincrementallowerbound},
that demonstrates that the $1.5$ factor from Theorem~\ref{thm:RePairToTheEnd} is tight,
assuming certain encodings, for $k=0$.
The construction employs large alphabets, i.e.~$\sigma = \Theta(|S|)$,
and as Theorem~\ref{thm:RePairToTheEnd} assumes $k = o(\log_\sigma n)$,
this means that we consider $k=0$.
While such large alphabets do not reflect practical cases,
in which $\sigma$ is much smaller than $|S|$,
still, in the case of grammar compression this example gives some valuable intuition:
replacing the substring $w$ decreases the size counted
in symbols but may not always decrease encoding size,
as we have to store additional information
regarding replaced string, which is costly for some encodings.

\begin{lemma}
	\label{lem:15entropyRePairincrementallowerbound}
	There exists a family of strings $S$ such that \RePair with
	both \cnfenc and entropy encoding
	uses at least $ \frac{3}{2}|S|H_0(S) - o(|S|\log \sigma )$ bits,
	assuming that we encode the grammar of size $g$ (i.e.\ with $g$ nonterminals)
	using at least $g \log g - \Ocomp(g)$ bits.
	Moreover, $|S|H_0(S) = \Omega(|S|\log \sigma )$,
	which implies that the cost of the encoding, denoted by $\RePair(S)$, satisfies
	$\limsup_{|S|\rightarrow \infty }
	\frac{\RePair(S)}{|S|H_0(S)} \geq \frac{3}{2}$.
\end{lemma}
The $n$-th string for Lemma~\ref{lem:15entropyRePairincrementallowerbound} 
is over an alphabet $\Gamma_n = \{a_1, a_2, \ldots, a_n, \# \}$
and is equal to $S_n = a_1 \#  a_2 \# a_3 \# \cdots a_{n-1} \# a_n \#$
	$a_n \# a_{n-1} \cdots a_2 \#  a_1 \#$.

The strings from~Lemma~\ref{lem:15entropyRePairincrementallowerbound}
show that at some iteration bit encoding of \RePair can increase.
Even though the presented example requires a large alphabet
and is somehow artificial,
we cannot ensure that a similar instance does not occur
at some iteration of \RePair, as the size of the alphabet of the working string increases.
In the above example the size of the grammar was significant.
This is the only possibility to increase bit size as by Theorem~\ref{thm:main_estimation}
adding new symbols does not increase entropy encoding of working string significantly.

The key observation in the proof of Theorem~\ref{thm:RePairToTheEnd},
which can be of independent interest,
is that for a large entropy strings
adding rules to a grammar does not increase
the total encoding cost, i.e.\ $|S|H_0(S) + \mathcal |N(G)| \log \mathcal |N(G)|$,  that much.
Intuitively, even though adding a rule to a grammar $G$ costs
about $|\mathcal N(G)| \log |\mathcal N(G)|$ bits by Lemma~\ref{lem:grammar_encodings},
if the string $S$ have entropy $H_0(S) \approx \log \mathcal |S|$,
then after replacing a pair with two occurrences 
$|S|H_0(S)$ will decrease by $\approx 2\log |S|$.
Thus the decrease in $|S|H_0(S)$ after replacing a pair
should offset the cost of the grammar in the encoding cost.
A similar, yet weaker, statement is true if $H_0(S) \geq (\log |S|)/\beta$:
some replacements may not decrease $|S|H_0(S)$
(as in the example given in Lemma~\ref{lem:15entropyRePairincrementallowerbound}),
but if we do enough replacements, the $|S|H_0(S)$ will eventually decrease
and in total it will offset the cost of encoding the grammar to some degree.
We show the special case of this observation for $\beta=2$,
i.e.\ if the entropy is $H_0(S) \geq \frac{|S|}{2}$,
then the encoding cost is bounded by $\frac 3 2 |S|H_0(S)$ for both
incremental and entropy coding.
The tools we use in the full proof show
that this is true not only for \RePair, but also for a wide range of grammar compressors
(see weakly non-redundant grammars,
paragraph~\ref{par:weakly_non_red_grammars}, in the Appendix).

In the proof of Theorem~\ref{thm:RePairToTheEnd}
we look at the working string $S'$
after $|S|/\log^{1+\epsilon}|S|$ \RePair's iterations.
Till this point the grammar $G_0$ must have at most $|S|/\log^{1+\epsilon}|S|$ rules,
thus encoding of those rules takes $o(|S|)$ bits so far.
Denote $S''$ as the final working string and $G'$ as the final grammar.
Then we estimate the encoding of grammar $G = G'\setminus G_0$ and
the encoding of $S''$, this can be viewed as if we would estimate the
size of grammar $(S'', G)$ obtained by running \RePair on $S'$.
If $|S'| \leq |S|/\log^{1+\epsilon}|S|$
then even the trivial encoding takes $o(|S|)$ bits.
Otherwise, we consider incremental and entropy coding separately,
first consider the case for incremental encoding.
Then the encoding size of $(S'', G)$
is upper bounded by roughly $|S''|H_0(S'') + \nonterminals{G} \log |S|$,
we show that this is bounded by $\frac{3}{2}|S'|H_0(S')$.
For this, we write $|S''|H_0(S'') + \nonterminals{G} \log |S|$
as $\frac 1 2 |S''|H_0(S'') + (\frac 1 2 |S''|H_0(S'') + \nonterminals{G} \log |S|)$.
As the first term is around $\frac 1 2 |S'|H_0(S')$ by Theorem~\ref{thm:main_estimation}
applied with $k=0$ ($S''$ induces a parsing of $S'$),
we need to estimate the $\frac 1 2 |S''|H_0(S'') + \nonterminals{G}\log |S|$ summand.
Note that $\nonterminals{G} \leq (|S'| - |S''|)/2$ as each nonterminal corresponds
to at least two replacements of a pair,
thus it is enough to estimate $|S''| \log |S|$ in terms of $|S'|H_0(S')$.
To this we use Theorem~\ref{thm:main_estimation} to lower bound the $|S'|H_0(S')$:
Consider a trivial parsing $Y_P$ of $S'$ into consecutive pairs of letters,
the calculation shows that we can lower bound the entropy of parsing
$|Y_P|H_0(Y_P)$ by roughly $|S'|/2 \log |S|$
(as at this point each pair occurs at most $\Ocomp(\log^{1+\epsilon} |S|)$ times).
Thus $|S'|H_0(S') \geq |Y_P|H_0(Y_P) \geq (|S'|/2) \log |S|$,
and together with $\nonterminals{G} \leq (|S'| - |S''|)/2$ this allows to
estimate the $(\frac 1 2 |S''|H_0(S'') + \nonterminals{G} \log |S|)
\leq \log |S|\cdot \left(\frac 1 2 |S''| + |\mathcal{N}(G)|\right)$ summand by $|S'|H_0(S')$.
At the end, we again use the Theorem~\ref{thm:main_estimation}
to show that $\frac 3 2 |S'|H_0(S') \leq |S|H_k(S) + o(|S|\log \sigma)$.
As we consider different grammars (for $S$ and $S'$),
some additional technical care is needed,
i.e.\ we must show that encodings of $G'$ and $G'\setminus G_0$ are asymptotically similar.
Arguing about the full entropy encoding is more involved and
requires some additional considerations,
but the main idea is similar.

\begin{theorem}
\label{thm:RePairToTheEnd}
	Let $S$ be a string over a $\sigma$-size alphabet,
	$k = o(\log_\sigma |S|)$.
	Then the size of \RePair{} output
	is at most $\frac{3}{2}|S| H_k(S) + o(|S| \log \sigma)$
	for \cnfenc and entropy encoding.
\end{theorem}

\subsection{Irreducible grammars and their properties}\label{subsec:irreducible}
Kieffer et al.\ introduced the concept of \emph{irreducible grammars}~\cite{kieffer2000grammar},
which formalise the idea that there is no immediate way
to make the grammar smaller.
Many heuristics fall into this category:
\Sequential, \Sequitur, \LongestMatch, and \Greedy,
even though some were invented before the notion was introduced~\cite{SmallestGrammar}.
They also developed an encoding of such grammars
which was used as a universal code for finite-state sources~\cite{kieffer2000grammar}.

\begin{definition}
	A full grammar $(S, G)$ is \emph{irreducible} if:
	\begin{enumerate}[({IG}1)]
		\item 
		no two distinct nonterminals have the same expansion; \label{IG1}\label{expuniq}
		\item 
		every nonterminal, except the starting symbol,
		occurs at least twice in $\grhs(S,G)$;
		\label{IG2}\label{nonred}
		\item 
		no pair occurs twice
		(without overlaps) in $S$ and $G$'s right-hand sides.
		\label{IG3}\label{pairuniq}
	\end{enumerate}
\end{definition}

Unfortunately, most irreducible grammars have the same drawback as \RePair:
they can introduce new symbols
without decreasing entropy of the starting string but increasing the bit size of the grammar.
In particular, the example from Lemma~\ref{lem:15entropyRePairincrementallowerbound} 
applies to irreducible grammars
(i.e.\ the grammar produced by \RePair in Lemma~\ref{lem:15entropyRePairincrementallowerbound}
is irreducible).

Ideally, for an irreducible grammar we would like to repeat the argument
used for \RePair:
we can stop at some iteration
(or decompress some nonterminals as in~\cite{ruleReductionSequential})
such that the grammar is small and has 
$\Ocomp(|S|^c)$ nonterminals, for some constant $c<1$. 
It turns out that there are irreducible grammars
which do not have this property,
see Lemma~\ref{lem:irreducibleDecompression}.
Moreover, grammar compressors that work in a top-down manner,
like \LongestMatch, tend to produce such grammars.
Furthermore, the grammar in Lemma~\ref{lem:irreducibleDecompression}
has size $\Theta\left(\frac{|S|}{\log_\sigma |S|}\right)$,
which is both an upper and lower bound
for the size of an irreducible grammar.

\begin{lemma}
\label{lem:irreducibleDecompression}
For a fixed $\sigma$ 
there exists a family of irreducible grammars $\{G_i\}_{i \geq 0}$, where $G_i$ generates a string $S_i$,
such that $|S_1| < |S_2| < \cdots$, size of
$G_i$ is $\Theta\left(|S_i| / \log_\sigma |S_i| \right)$
and for every constant $0 < c < 1$ decompressing any set of nonterminals so that 
only $\Ocomp(|S_i|^c)$ nonterminals remain yields a~grammar
of size $ \Theta(|S|)  =  \omega(|S_i|/\log_\sigma |S_i|)$,
i.e.\ the grammar is no longer small.
\end{lemma}
Moreover,
some encodings of such grammars can give larger output than~\RePair's,
i.e.\ the $1.5$ bound does not hold for irreducible grammars.
\begin{lemma}
\label{lem:20irreduciblelowerbound}
There exists a family of strings $\{S_i\}_{i=1}^\infty$ such that for each string $S_i$
there exists an irreducible grammar $G_i$ whose entropy coding takes
at least $2|S_i|H_0(S_i) - o(|S_i|\log\sigma)$ bits.
Moreover $|S_i|H_0(S_i) = \Omega(|S_i|\log \sigma )$,
i.e.\ the cost of the encoding, denoted by $\algofont{Irreducible}(S)$, satisfies
$\limsup_{i \rightarrow \infty }
\frac{\algofont{Irreducible}(S_i)}{|S_i|H_0(S_i)} \geq 2$.
\end{lemma}

Still, we give upper bounds for encodings of irreducible grammars, 
with worse constants though.
The rough idea of the proof of Theorem~\ref{thm:irreducible_entropy_concatenation}
is as follows:
For an irreducible grammar $(S', G)$
we consider the concatenation $S_G$ of the strings in $\grhs{(S', G)}$.
Then we show that $S_G$
can be obtained by permuting letters of two strings $S''$ and $S_\mathcal{N}$
where: $S''$ is a parsing of $S$
and $S_\mathcal{N}$ can be obtained by permuting and removing letters
from $S''$.
Thus $|S_G|H_0(S_G) = |S''S_\mathcal{N}|H_0(S''S_\mathcal{N}) \leq 2|S''|H_0(S'')$,
by applying Theorem~\ref{thm:main_estimation}
on $S''$ we get that $2|S''|H_0(S'')$ is bounded by $2|S|H_k(S)  + o(|S|\log \sigma)$.

The proof of Theorem~\ref{thm:irreducible_naive},
similarly to the proof of Theorem~\ref{thm:RePairToTheEnd},
uses Theorem~\ref{thm:main_estimation} to lower bound
the entropy of $S_G$:
In irreducible grammar, each pair occurs only once 
in strings of  $\grhs{(S', G)}$, thus we can find a parsing $Y_{S_G}$ of $S_G$
into pairs and letters where $|S_G|/3$ phrases are unique
(note that all rules can be of length $3$).
Thus the zeroth-order entropy of the parsing is at least
$|Y_{S_G}|H_0(Y_{S_G}) \geq (|S_G|/3) \log (|S_G|/3)$.
Then by Theorem~\ref{thm:main_estimation} we get that 
$(|S_G|/3) \log (|S_G|/3)\leq |Y_{S_G}|H_0(Y_{S_G})  - \Ocomp(|Y_{S_G}|)\leq |S_G|H_0(S_G)$,
plugging the estimation on $|S_G|H_0(S_G)$
from the Theorem~\ref{thm:irreducible_entropy_concatenation}
yields the claim.
%
%
\begin{theorem}
	\label{thm:irreducible_entropy_concatenation}
	Let $S$ be a string over $\sigma$-sized alphabet, $k = o(\log_\sigma |S|)$.
	Then the size of the entropy coding of any irreducible full grammar
	generating $S$ is at most $2 |S| H_k(S) + o(|S| \log \sigma)$.
\end{theorem}
\begin{theorem}
	\label{thm:irreducible_naive}
	Let $S$ be a string over $\sigma$-sized alphabet,
	$k = o(\log_\sigma |S|)$.
	The~size of the fully naive coding of any irreducible full grammar
	generating $S$ is at most~$6 |S| H_k(S) + o(|S| \log \sigma)$.
\end{theorem}

\subsection{\Greedy}\label{sec:greedy}
\Greedy~\cite{apostolico1998some} can be viewed as non-binary \RePair:
in each round, it replaces a substring occurring in some strings of $\grhs(S,G)$,
obtaining $(S',G')$, such that $\rhs{S',G'}$ is the smallest possible.
For this, at each round, for substrings of strings in $\grhs(S,G)$
it calculates the maximum number of non-overlapping occurrences.
This can be realised with creating a suffix tree at each iteration,
and so its asymptotic construction time has been bounded
only by $\Ocomp(|S|^2)$ so far.
Note that, contrary to \RePair, it can correctly estimate
the decrease in the number of symbols when replacing pair $AA$ in a maximal substring $A^n$.
It is known to produce small grammars in practice, both in terms of nonterminal
size and bit size.
Moreover, it is notorious for being hard to analyse in terms of approximation ratio~\cite{SmallestGrammar,SLPaproxrevisited}.

\Greedy has similar properties as \RePair:
after each iteration
the frequency of the most frequent pair on the right-hand side of the grammar
does not increase 
and the size of the grammar can be estimated in terms
of this frequency. 
In particular, there always is a point
of its execution in which the number of nonterminals is $\Ocomp(|S|^c), c < 1$ and the full grammar is of size $\Ocomp\left(|S|/\log_\sigma |S| \right)$.
The entropy encoding at this time yields $|S|H_k(S) + o(|S|\log \sigma)$ bits,
while \Greedy run till the end achieves $\frac3 2 |S|H_k(S)$ using entropy coding and $2|S|H_k(S)$ using fully naive encoding,
so the same as in the case of \RePair.


In practice stopping \Greedy early is beneficial:
similarly as in \RePair there can exist a point where
we add new symbols that do not decrease the bitsize of the output.
Indeed, it was suggested~\cite{apostolico1998some}
to stop \Greedy as soon as the decrease in the grammar size is small enough,
yet this was not experimentally evaluated.
Moreover, as the time needed for one iteration is linear 
we can stop after $\Ocomp(|S|^c)$ iterations
ending with $\Ocomp(|S|^{1+c})$ running time,
again our analysis suggests that $c$ factor depends on the constant
in additional summand.


The proofs for \Greedy are similar in spirit to those for \RePair,
yet technical details differ.
In particular, we use the (non-trivial) fact that \Greedy produces irreducible grammar
and at every iteration conditions \IGreftwo{1}{2} are satisfied, i.e.
that distinct nonterminals have distinct expansions
and that every nonterminal is used at least twice.
All of those properties are a consequence of \Greedy being a \emph{global} algorithm,
see~\cite{SmallestGrammar}.

\begin{theorem}
	\label{GreedyStopped}
	Let $S$ be a string over $\sigma$-sized alphabet and $k = o(\log_\sigma |S|)$.
	Let $\xi >0, \beta > 2$ be constants.
	Let $(S', G)$ be full grammar generating $S$ after some iteration of \Greedy.
	When the $\rhs{S', G}$ is first below 
	$\rhs{S', G} < \frac{ (2\xi + 10\beta) |S|}{\log_\sigma |S|} + 8|S|^\frac{2}{\beta} + 2$ 
	the number of nonterminals $\nonterminals{G}$ of $G$ is at most
	$1 + \frac{2}{\xi} |S|^\frac{2}{\beta} \log_\sigma |S| $;
	for every $\xi > 0, \beta > 2$ and $S$ such a point always exists.
	Moreover, at this moment the bit-size of entropy coding
	of $(S', G)$ is at most $|S|H_k(S) + o(|S| \log \sigma)$.
\end{theorem}
\begin{theorem}
	\label{thm:greedy_naive}
	Let $S$ be a string over $\sigma$-sized alphabet and $k = o(\log_\sigma |S|)$.
	Then the size of the fully naive encoding of 
	full grammar produced by \Greedy
	is at most $2|S|H_k(S) + o(|S| \log \sigma)$; 
	moreover the size of the entropy encoding of
	full grammar produced by \Greedy is at most $\frac{3}{2} |S|H_k(S) + o(|S| \log \sigma)$.
\end{theorem}

\section{Entropy bound for string parsing}
\label{sec:entropy_bound_for}
We now prove Theorem~\ref{thm:main_estimation} and Theorem~\ref{thm:mean_entropy}.
We make a connection between
the entropy of the parsing of a string $S$, i.e.\ $|Y_S|H_0(Y_S)$,
and the $k$-order empirical entropy $|S|H_k(S)$ of $S$;
this can be seen as a refinement and strengthening of results that relate
$|Y_S|\log|Y_S|$ to $H_k(S)$~\cite[Lem.~2.3]{KosarajuManzini99}.
To compare the two results,
our result establishes upper bounds when phrases
of $Y_S$ are encoded using an entropy coder
while previous ones apply to the trivial encoding of $Y_S$,
which assigns to each letter $\log |Y_S|$ bits.

The proofs follow in a couple of simple steps.
First, 
we recall a strengthening of the fact that entropy
lower-bounds the size of the encoding of the string using any prefix codes:
instead of assigning natural lengths (of codes) to letters of $Y_S$
we can assign them some probabilities
(and think that $-\log(p)$ is the ``length'' of the prefix code).
Then we define the probabilities of phrases in the parsing $Y_S$
in such a way that they relate to higher-order entropies.
Depending on the result we either look at fixed, $k$-letter
context or $(i-1)$-letter context for the $i$-th letter of a phrase.
%
%
This already yields the first claim of Theorem~\ref{thm:main_estimation},
to obtain the second we substitute the estimation on the
parsing size and make some basic estimation on the entropy of the lengths,
which is a textbook knowledge~\cite[Lemma 13.5.4]{ElementfofInformationTheory2006}.



\begin{lemma}[\cite{aczel1973shannon}]\label{theoremP}
	Let $w$ be a string over alphabet $\Gamma$
	and $p:\Gamma \rightarrow \mathbb{R}^+$ be a function such that $\sum_{s \in \Gamma} p(s)  \leq 1$.
	Then
	$
	|w|H_0(w) \leq -\sum_{s \in \Gamma} |w|_{s} \log p(s)
	$.
\end{lemma}

\label{sec:parsings-costs}
To use Lemma~\ref{theoremP}
we need an appropriate valuation of $p$ for phrases of the parsing.
We assign $p$-value to each individual letter in each phrase
and then for a phrase $y$, $p(y)$ is a product of values of consecutive letters of $y$.

In the case of Theorem~\ref{thm:mean_entropy} 
for the $j$-th letter of a phrase $y$
we assign the probability of this letter occurring in $j-1$
letter context in $S$,
i.e.\ 
by Definition~\ref{def:k_order_entropy} in most cases
$\prob(y[j] \ | \ y[1\ldots j-1]) = \frac{|S|_{y[1\ldots j]}}{|S|_{y[1\ldots j-1]}}$,
i.e.\ we think that we encode the letter using
$(j-1)$-th order entropy.
The difference in the case of Theorem~\ref{thm:main_estimation},
is that we assign $\frac{1}{\sigma}$
to the first $k$ letters of the phrase,
the remaining ones are assigned 
probability of this letter occurring in $k$-letter context.
The phrase costs are simply logarithmed values of phrase probabilities,
which intuitively can be viewed as the cost of bit-encoding of a phrase.

\begin{definition}[Phrase probability, parsing cost]
	\label{def:phrasecost}
	Given a string $S$ and its parsing $Y_S = y_1 y_2 \ldots y_c$
	the \emph{phrase probability} $\mathcal P(y_i)$
	and \emph{$k$-bounded phrase probability} are:
	\begin{equation*}
	\mathcal{P}(y_i) = \prod _{j=1}^{|y_i|} \prob(y_i[j] \ | \ y_i[1\ldots j-1])
	\quad \text{and} \quad	 
	\mathcal{P}_k(y_i) =  \frac{1}{\sigma^{\min(|y_i|,k)}} \prod _{j=k+1}^{|y_i|} 
	\prob(y_i[j] \ | \ y_i[j-k\ldots j-1]) \enspace ,
	\end{equation*}
	where $\prob(a | w)$ is the empirical probability of letter $a$ occurring
	in context $w$, cf.~Definition~\ref{def:k_order_entropy}.

	%
	
	The \emph{phrase cost} and \emph{parsing cost} are
	$C(y_i) = 
	-\log \mathcal{P}(y_i)$ 	
	and
	$C(Y_S) = \sum_{i = 1}^{|Y_S|} C(y_i)$.
	Similarly the \emph{$k$-bounded phrase cost} and \emph{$k$-bounded parsing cost} are:
	$C_k(y_i) = 
	- \log \mathcal{P}_k(y_i)$ and 
	$C_k(Y_S) = \sum_{i = 1}^{|Y_S|} C_k(y_i)$.
\end{definition}


When comparing the $C_k(Y_S)$ cost and $|S|H_k(S)$,
the latter always uses $-\prob(a | w)$  bits for each symbol $a$ occurring in context $w$,
while $C_k(Y_S)$ uses $\log \sigma$ bits on each of the first $k$ letters
of each phrase, thus intuitively it loses up to $\log \sigma$ bits
on each of those $|Y_S|k$ letters.

\begin{lemma}
	\label{lem:cost_and_kentropy}
	Let $S$ 
	be a string and let $Y_S$ be its parsing. Then
	for any $k$:
	$C_k( Y_S )   \leq |S|H_k(S) + |Y_S|k\log \sigma$.
	%
	%
	%
\end{lemma}

\begin{lemma}
	\label{lem:mean_entr_cost}
	Let $S$ be a string.
	Then for any $l$ there exists a parsing $Y_S$ of length $|Y_S| \leq \left\lceil |S|/l\right\rceil + 1$ such that
	each phrase except first and last has length exactly $l$ and 
	$C(Y_S) \leq  \frac{|S|}{l}\sum_{i=0}^{l-1} H_i(S)+\log |S|$.
\end{lemma}

Ideally, we would like to plug in the phrase probabilities for $Y_S$
into Lemma~\ref{theoremP}
and so obtain that the parsing cost $C(Y_S)$ upper-bounds
entropy coding of $Y_S$, i.e.\ $|Y_S| H_0(Y_S)$.
But the assumption of Lemma~\ref{theoremP} 
(that the values of function $p$ sum to at most $1$)
may not hold as we can have phrases of different lengths
and so their probabilities can somehow mix.
Thus we also take into account the lengths of the phrases:
we multiply the phrase probability by the probability of $|y|$,
i.e.\ the frequency of $|y|$ in $\Lengths(Y_S)$. 
After simple calculations, we conclude that $|Y_S|H_0(Y_S)$
is upper bounded by the parsing cost $C(Y_S)$ plus the zeroth-order entropy
of phrases' lengths: i.e.\ when $L = \Lengths(Y_S)$, the $|L| H_0(L)$.


\begin{theorem}\label{thmHC}
	Let $S$ be a string over $\sigma$-size alphabet, $Y_S$ its parsing,
	$c = |Y_S|$ and $L = \Lengths(Y_S)$. Then:
	\begin{equation*}
	|Y_S| H_0(Y_S)
	\leq
	C(Y_S) + |L| H_0(L)  \quad \text{and} \quad 
	|Y_S| H_0(Y_S)
	\leq
	C_k(Y_S) + |L| H_0(L) \enspace .
	\end{equation*}
\end{theorem}
When $S \in a^*$ (and so $C(S) = C_k(S) = H_0(S) = 0$)
the entropy $|Y_S| H_0(Y_S)$ \emph{is} the entropy of $|L| H_0(L)$,
thus summand $|L| H_0(L)$ on the right-hand is necessary.





By textbook arguments~\cite[Lem.~13.5.4]{ElementfofInformationTheory2006},
the entropy of lengths is always
$\Ocomp(|S|)$, and
for small enough $|Y_S|$ it is $\Ocomp\left(|S| \log \log_\sigma |S|/\log_\sigma |S|\right)$;
moreover in the case of Theorem~\ref{thm:mean_entropy}
it can be bounded by $\Ocomp(\log |S|)$.

Theorems~\ref{thm:main_estimation}
and~\ref{thm:mean_entropy} follow directly from Lemmas~\ref{lem:cost_and_kentropy}, \ref{lem:mean_entr_cost},
 Theorem~\ref{thmHC} and above estimations.



\bibliography{references}

\clearpage
\appendix

\begin{center}
{\Huge \sf Appendix}
\end{center}

\section{Additional material for Section~\ref{sec:definitions}}
For simplicity, if $i>j$ then $w[i\ldots j] = \epsilon$.

We extend the notion $|w|_v$ to sets of words, i.e.\ $|w|_V = \sum_{v \in V} |w|_v$.

The \emph{size} of the nonterminal $\rhs{X}$ is the length
of its right-hand side,
for a grammar $G$ or full grammar $(S,G)$
the $\rhs{G}, \rhs{S,G}$ denote the sum of lengths
of strings in $\grhs(G), \grhs(S,G)$, respectively.

\section{Expanded commentary for Section~\ref{sec:upper-bounds-on-grammar-compression}}

We prove that \RePair and \Greedy stopped after certain iteration achieve $|S|H_k(S) + o(|S| \log \sigma)$ bits
and their number of different nonterminals is small, i.e\ $\Ocomp(|S|^c)$ for some $c < 1$.
Our analysis shows that we can decrease value $c$ at the
expense of the constant hidden in $o(|S| \log \sigma)$.
This is particularly important in practical applications
as it translates to a grammar which has small output alphabet and,
as in most compressors every symbol is encoded using prefix-free code,
this reduces the bits required to store this code.
In comparison, certain dictionary methods, like \lzso,
do not have this property, as they can produce string
over a~large, i.e\ $\Omega({|S|}/{\log_\sigma |S|})$-sized alphabet,
which implies a large dictionary size.

\subsection{Proofs for Section~\ref{subsec:encoding_of_grammars} (Encoding of grammars)}
We first complete the promised argument that at least some of
the original encodings of \RePair match previously presented bound on CNF encoding.
\begin{lemma}
\label{lem:originalRepairEncodingBound}
\RePair{} original grammar encoding together with an arithmetic coding
encodes a grammar $G$ using $\nonterminals{G} \log (\nonterminals{G} + \sigma)
+ \Ocomp(\nonterminals{G} + \sigma)$ bits.
\end{lemma}
\begin{proof}
As mentioned before,
one of the encodings originally used for \RePair~\cite{RePair}
is based on division of nonterminals $G$ into groups
$g_1, \ldots, g_{z}$, where $X\rightarrow AB \in g_i$ if and only if
$A \in g_{i-1}$ and $B\in g_j$ or $B \in g_{i-1}$ and $A \in g_j$,
for some $j \leq i-1$.
To complete the definition, $g_0$ contains all letters of the input alphabet.
In short, group $g_i$ contains non-terminals which are at height $i$
in the derivation tree of $G$.
It was shown~\cite{RePair} that each group $g_i, \ i > 0$, 
can be represented as a bitvector of size $p_{i-1}^2 -p_{i-2}^2$
in which $p_i - p_{i-1}$ entries are ones,
where $p_i = (\sum_{j\leq i} |g_j|)$.
Those arrays were then encoded using arithmetic coding~\cite{RePair}.
When probabilities
$P(1) = \frac{1}{\nonterminals{G} + \sigma}$
and $P(0) = \frac{\nonterminals{G} + \sigma -1}{\nonterminals{G} + \sigma}$
are used in such coding,
each group is encoded using at most 
$(p_i - p_{i-1}) \log (\nonterminals{G} + \sigma)
 +
(p_{i-1}^2 -p_{i-2}^2) \log \frac{\nonterminals{G}+\sigma}{\nonterminals{G}+\sigma-1}  + \Ocomp(1)$ bits.
Summing over all groups we get that the size is bounded by
\begin{equation*}
\sum_{i=1}^z \left( (p_i - p_{i-1}) \log(\nonterminals{G} + \sigma) +
(p_{i-1}^2 -p_{i-2}^2) \log \frac{\nonterminals{G} + \sigma }{\nonterminals{G}+\sigma-1}  + \Ocomp(1)  \right)
\end{equation*}
As the sum telescopes, this can be bounded by:
\begin{align*}
&
(p_z - p_0)\log (\nonterminals{G} + \sigma) +
(p_{z-1}^2 - p_0^2) \log \frac{\nonterminals{G} + \sigma}{\nonterminals{G} + \sigma-1} + \Ocomp(z)\\ 
 &\leq
\nonterminals{G}\log (\nonterminals{G} + \sigma) +
(\nonterminals{G} + \sigma) ^2 \log \frac{\nonterminals{G} + \sigma}{\nonterminals{G} + \sigma-1}  + \Ocomp(\nonterminals{G} )
\enspace ,
\end{align*}
when the inequality follows as $p_z = \nonterminals{G} + \sigma$ and $p_0 =\sigma$.
As $x \log \frac{x}{x-1} = \Ocomp(1)$ we have that the total cost of encoding is at most:
\begin{equation*}
\nonterminals{G} \log (\nonterminals{G} + \sigma)
+ \Ocomp(\nonterminals{G} + \sigma) \enspace . \qedhere
\end{equation*}
\end{proof}
It is worth mentioning that instead of using fixed arithmetic coder original
solution used an adaptive one, which can be beneficial in practice.

\begin{proof}[Proof of Lemma~\ref{lem:grammar_encodings}]
	In each of the encodings we need to store sizes of 
	the right-hand sides of the nonterminals,
	storing these values in unary gives sufficient bounds
	for each method, i.e.~it consumes $\Ocomp(\rhs{S',G})$ bits.
	In the case of entropy coding, it is sometimes more practical
	to add separator characters to the concatenation of rules' right-hand sides,
	observe that this adds extra $\Ocomp(\rhs{S',G} )$ bits.
	To see this observe that if we have string $S \in \Sigma$ encoded with
	some prefix code $p:\Sigma \to \{0,1\}^*$
	in a way that the encoding uses $B$ bits,
	then we can devise the new code $p':(\Sigma \cup \#) \to \{0,1\}^*$
	which encodes $w$ with $B + |w|$
	bits and supports new character $\#$ by simply setting $p'(\#) = 0$
	and $p'(\sigma) = 1 p(\sigma)$, for $\sigma \neq \#$.
	Thus adding $f$ occurrences of a new letter to the string $w$ 
	increases encoding size by at most $|w| + f$.
	As in our case $ f < |w| = \rhs{S',G}$ the claim holds.
	
	The bound for fully naive encoding is obvious,
	as symbols on the right-hand sides are either one of the original $\sigma$ symbols or one of $\nonterminals{G}$ nonterminals.
	
	
	In the naive encoding the starting string is encoded
	using entropy coder (e.g.\ Huffman),
	the grammar is encoded in the same way as in the fully naive encoding.

	For the case of incremental encoding, let us first argue
	that we can always reorder nonterminals appropriately.
	Consider the following procedure:
	start with sequence $R$ which contains letters of the original alphabet in an arbitrary order.
	Through the procedure the sequence $R$ will contain already processed letters and nonterminals.
	We append nonterminals to $R$ in the following way:
	take the first unprocessed nonterminal $X\rightarrow YZ$
	with $Y$ being the earliest possible nonterminal or letter occurring in $R$,
	and append $X$ to $R$.
	At the end of the above procedure it is necessary to 
	rename nonterminals so that they will correspond to numbers
	in created sequence.
	
	In incremental encoding we encode the second element of pairs naively.
	As the first elements are sorted, we store only consecutive differences.
	For this differences we use Elias $\delta$-codes.
	These codes, for a number $x$ consume at most $\log x + 2 \log \log(1 +  x) + 1$ bits.
	As all these differences sum up to $\nonterminals{G} + \sigma$,
	it can be shown that this encoding takes
	at most $\Ocomp(\nonterminals{G} + \sigma)$ bits~\cite[Theorem 2.2]{sadakane2006squeezing}).
	Observe that encoding these differences in unary is also within
	the required bounds.
\end{proof}

\begin{lemma}
	\label{lem:huffman_entropy_coder}
	Let $S$ be a string over alphabet $\Sigma \subseteq \{1, \ldots, |S|\}$.
	Then there is a Huffman-based encoder which uses at most $|S|H_0(S) + \Ocomp(S)$ bits
	to encode $S$ (together with a dictionary).
\end{lemma}
\begin{proof}
	It is a well-known fact, that Huffman encoding of a string $S$ takes at most
	$|S|H_0(S) + |S|$ bits.
	Observe that we do not need to store the dictionary explicitly:
	assuming that we use a deterministic procedure to build the Huffman tree we only need to
	store frequencies of each letter.
	It is sufficient to store the frequencies in unary, i.e.\ we store
	string $1^{f_1}01^{f_2}0\ldots01^{f_{|S|}}$,
	where $f_i$ is a frequency of $i$-th letter, note that some $f_i$s could be $0$.
	Storing the frequencies takes $2|S|-1 = \Ocomp(|S|)$ bits.
\end{proof}

\begin{lemma}[Full version of Lemma~\ref{lem:entropy_coding_concatenation}]
	Let $S$ be a string over an alphabet of size $\sigma$
	and $(S', G)$ a full grammar generating it,
	and $k = o(\log_\sigma |S|)$.
	Let $S_G$ be a concatenation of strings in $\grhs(S',G)$.
	Assume that no two nonterminals in $(S', G)$ have the same expansion.
	If $|S_G| = \Ocomp\left( \frac{|S|}{\log_\sigma |S|} \right)$ then
	\begin{equation*}
	|S_G| H_0(S_G) \leq |S|H_k(S) + \nonterminals{G}\log |S| +  o(|S| \log \sigma) \enspace .
	\end{equation*}
	In particular, if additionally $\nonterminals{G} = o\left(\frac{|S|}{\log_\sigma |S|}\right)$ then
	\begin{equation*}
	|S_G| H_0(S_G) \leq |S|H_k(S) + o(|S| \log \sigma) \enspace .
	\end{equation*}
	The same bounds apply to the entropy coding of $(S',G)$.
\end{lemma}

\begin{proof}
	From the full grammar $(S', G)$
	we create a parsing $Y_{S}$ of string $S$ of size
	$|Y_S| \leq \rhs{G} + |S'|$
	by the following procedure:
	Take the starting string $S'$ as $S''$.
	While there is a nonterminal $X$ such that it occurs in $S''$
	and it was not processed before,
	replace one of its occurrences in $S''$
	with the $\grhs(X)$.
	Clearly, $S''$ is over the input alphabet and nonterminals and
	$|S''| = |S'| + \rhs{G} - \nonterminals{G} = |S_G| - \nonterminals{G} 
	= \Ocomp\left(\frac{|S|}{\log_\sigma |S|}\right)$,
	as each right-hand side is substituted and each nonterminal
	is replaced once.
	String $S''$ induces the parsing $Y_S$ of $S$
	and applying 
	Theorem~\ref{thm:main_estimation} to $Y_S$
	yields
\begin{equation}
	\label{eq:lem_entropy_concatenation_parsing_bounds}
	|Y_{S}|H_0(Y_{S}) \leq |S|H_k(S) + o(|S| \log \sigma) \enspace.
\end{equation}
	
	It is left to compare $|Y_{S}|H_0(Y_{S})$
	with the entropy of $S_G$, i.e.\ $|S_G|H_0(S_G)$.
	Note that up to permuting of letters, $S''$ is a concatenation
	of $S'$ and $\grhs{(G)}$ with one occurrence of each nonterminal removed.
	From the definition of entropy, adding one letter to string 
	increases encoding cost by at most logarithmic factor, i.e.\
	for any string $w$ and letter $a$ it holds that:
	\begin{align*}
	&|wa|H_0(wa) - |w|H_0(w) \\
	&=
	\sum_{b \in \Sigma_{wa}} |wa|_b \log \frac{|wa|}{|wa|_b} - 
	\sum_{b \in \Sigma_{w}} |w|_b \log \frac{|w|}{|w|_b}\\
	 &=
	\sum_{b \in \Sigma_{wa}\setminus\{a\}} |w|_b
	\left( \log \frac{|w|+1}{|w|_b} - \log \frac{|w|}{|w|_b}  \right)
	 + |wa|_a \log \frac{|wa|}{|wa|_a} - |w|_a\log \frac{|w|}{|w|_a} & \enspace
	\text{as $|wa|_a = |w|_b$}\\
	&= \sum_{b \in \Sigma_{wa}\setminus\{a\}} |w|_b \log \frac{|w|+1}{|w|}
	+ |w|_a \left(
	\log \frac{|w|+1}{|wa|_a} - \log \frac{|w|}{|w|_a} \right) + \log \frac{|w|+1}{|wa|_a}\\
	&\leq |w| \log \frac{|w|+1}{|w|} + \log (|w|+1)\\
	&\leq  \log (|w|+1) + \log \mathrm{e} = \log (|w|+1) + \Ocomp(1) \enspace ,
	\end{align*}
	where we assumed that the $|w|_a \log (\cdot)$ summand is $0$ if $|w|_a = 0$.
	Thus we obtain the bound on the entropy of $S_G$:
	\begin{align*}
	|S_G|H_0(S_G) &\leq
	|Y_S|H_0(Y_S) + \nonterminals{G}\log|S_G| + \Ocomp(\nonterminals{G})\\
	&\leq |S|H_k(S) + \nonterminals{G} \log |S| + \Ocomp(\nonterminals{G})  + o(|S| \log \sigma) \enspace . &\text{by~\eqref{eq:lem_entropy_concatenation_parsing_bounds} and as
	$|S_G| = \Ocomp(|S|)$
	 }
	\end{align*}
	The additional assumption on $\nonterminals{G}$ 
	in the second statement yields the second bound.
	The bounds on the entropy encoding follow from the fact that,
	by Lemma~\ref{lem:huffman_entropy_coder}, Huffman-based encoder
	uses at most $|S_G|H_0(S_G) + \Ocomp(|S_G|)$ bits for a string $S_G$.
\end{proof}
\subsection{Proofs for Section~\ref{subsec:repair} (\RePair)}

We 
state the
properties of \RePair,
some were used in the previous analysis~\cite{RePairentropy}:

\begin{lemma}[{\cite[Lem.~3]{RePairentropy}}]
	\label{lem:Repairfrequencedoesnotdecrease}
	The frequency of the most frequent pair in the working string does not increase
	during \RePair's execution.
\end{lemma}

\begin{lemma}
	\label{lem:most_frequent_and_number_of_nonterminals}
	If the most frequent pair in the working string occurs at least $z$ times
	then $\nonterminals{G} < \frac{2|S|}{z}$.
\end{lemma}
\begin{proof}
	By Lemma~\ref{lem:Repairfrequencedoesnotdecrease} all replaced pairs
	had frequency at least $z$ thus each such a replacement
	removed at least $z/2$ letters from the starting string
	(note that some occurrences could overlap).
	As it initially had $|S|$ letters, this cannot be done $|S|/(z/2)$ times.
\end{proof}

We also need the property that different nonterminals expand to different
substrings as we want to apply Theorem~\ref{thm:main_estimation}
to a working string, i.e.\ we need the fact that the working string
induces a parsing.
This can be proved by simple induction,
\cite[Lem.~4]{RePairentropy} gives a more general statement.
\begin{lemma}[{\cite[Lem.~4]{RePairentropy}}]
	\label{lem:repair_different_expansions}
	Let $G$ be a grammar generated by \RePair at any iteration.
	Then no two nonterminals $X, Y \in G$ have the same expansion.
\end{lemma}

Equipped with the above properties we are ready to prove the Theorem~\ref{thm:RePairStopped}.


\begin{proof}[proof of Theorem~\ref{thm:RePairStopped}]
We first prove the first part.
Assume that the working string $S'$ of \RePair has size
$c \geq \frac{(\xi + 3 \beta)|S|}{\log_\sigma |S|} + 2|S|^\frac{2}{\beta} + 1$.
Then $S'$ induces a parsing $Y_S = y_1 y_2 \ldots y_c$,
as by Lemma~\ref{lem:repair_different_expansions} different non-terminals expand
to different phrases.
A phrase $y$ in $Y_S$ is \emph{long} if $|y| \geq \frac{\log_\sigma |S|}{\beta} - 1$,
and \emph{short} otherwise.
Let $c_l$ be the number of long phrases.
Then
\begin{equation*}
c_l \cdot \left( \frac{\log_\sigma |S|}{\beta} - 1 \right)
\leq |S| \quad \implies \quad 
c_l
\leq
\frac{\beta \cdot |S|}{\log_\sigma |S| - \beta}  \enspace .
\end{equation*}

If $\beta > \frac{1}{3}\log_\sigma |S|$ then $c > |S|$,
thus the first part of the Theorem trivially holds.
So assume that $\beta \leq \frac{1}{3}\log_\sigma |S|$.
There are exactly $c-1$ digrams in $Y_S$.
Any long phrase can partake in $2$ pairs,
so the number of remaining pairs, which consist of two short phrases, is at least:
\begin{align}
\label{eq:short_phrases_lower_bound}
c - 1 - 2\cdot c_l 
&\notag\geq
\left(\frac{(\xi+3\beta)\cdot |S|}{\log_\sigma |S|} + 2|S|^\frac{2}{\beta} + 1 \right)- 1 - \frac{2\beta \cdot |S|}{\log_\sigma|S| - \beta}\\
&\notag\geq
\frac{(\xi+3\beta)\cdot |S|}{\log_\sigma |S|}
- \frac{2\beta \cdot |S|}{\log_\sigma|S| - \frac{1}{3}\log_\sigma |S|}  + 2|S|^\frac{2}{\beta}
\enspace & \text{as $\beta \leq \frac{1}{3} \log_\sigma |S|$}\\
&\notag>
\frac{(\xi+3\beta)\cdot |S|}{\log_\sigma |S|}
- \frac{3\beta \cdot |S|}{\log_\sigma |S|} + 2|S|^\frac{2}{\beta}\\
&\geq
\frac{\xi |S|}{\log_\sigma |S|} + 2|S|^\frac{2}{\beta}
\end{align}

On the other hand, the number of different short phrases is at most:
\begin{equation}
\label{eq:phrases_upper_bound}
\sum_{ i = 1}^{\left\lfloor \frac{\log_\sigma |S|}{\beta} - 1\right\rfloor }  \sigma^i
< \sigma^\frac{\log_\sigma |S|}{\beta} = |S|^\frac{1}{\beta} \enspace,
\end{equation}
as it is a geometric sequence.
As there are at most $|S|^\frac{1}{\beta} \cdot |S|^\frac{1}{\beta} = |S|^\frac{2}{\beta}$ different pairs consisting of short phrases,
there must exist a pair occurring at least:
\begin{equation}
\label{eq:eq_freq_repair}
\frac{\frac{\xi |S|}{\log_\sigma |S|} + 2|S|^\frac{2}{\beta} }{|S|^\frac{2}{\beta}}
\geq \frac{\xi \cdot |S|^{1-\frac{2}{\beta}}}{\log_\sigma |S|} + 2\enspace
\end{equation}
times, which is larger than $2$.
Note that this, in particular, implies that
if the working string is as long as $c$
then there is another iteration of \RePair.
Since there is a pair with frequency of value in~\eqref{eq:eq_freq_repair}
by Lemma~\ref{lem:most_frequent_and_number_of_nonterminals}
we can estimate the number of nonterminals at this point, $\nonterminals{G}$, in grammar $G$ as:
\begin{equation}
\label{eq:repair_stopped_grammar_size}
\nonterminals{G} < \frac{2|S|}{
	\frac{\frac{\xi |S|}{\log_\sigma |S|} + 2|S|^\frac{2}{\beta}}{|S|^\frac{2}{\beta}}}
\leq \frac{2}{\xi}|S|^\frac{2}{\beta} \cdot \log_\sigma |S| \enspace .
\end{equation}
Now consider the first iteration when $|S'| <
\frac{(\xi + 3 \beta)|S|}{\log_\sigma |S|} + 2|S|^\frac{2}{\beta} + 1$.
By the above reasoning, at this point the grammar must be of size at most
$\frac{2}{\xi}|S|^\frac{2}{\beta} \cdot \log_\sigma |S| + 1$,
this finishes the proof of the first part of the Theorem.

Applying Theorem~\ref{thm:main_estimation} we
can bound the entropy of the working string $S'$ by:
\begin{equation}
\label{eq:repair_stopped_string_size}
|S'|H_0(S') \leq |S|H_k(S) + o(|S|\log \sigma) \enspace ,
\end{equation}
as $|S'| = \Ocomp\left(\frac{|S|}{\log_\sigma |S|}\right)$.

Combining~\eqref{eq:repair_stopped_grammar_size}
with Lemma~\ref{lem:grammar_encodings} gives that 
the size of the naive and the incremental encoding
of full grammar $(S', G)$ is bounded by $|S'|H_0(S') + o(|S|\log \sigma)$,
note that we use the assumption that $\beta > 2$.
Together with~\eqref{eq:repair_stopped_string_size} this gives
that the encoding size is bounded by
$|S|H_k(S) + o(|S|\log \sigma)$ for both
the naive and the incremental encoding.
Finally, the bound for the entropy coding of full grammar $(S', G)$
follows directly from Lemma~\ref{lem:entropy_coding_concatenation},
as $\rhs{S', G}= |S'| + 2\nonterminals{G} = \Ocomp\left(\frac{|S|}{\log_\sigma S}\right)$.
\end{proof}

\begin{proof}[Proof of Lemma~\ref{lem:15entropyRePairincrementallowerbound}]
	Fix $n$ and an alphabet
	$\Gamma = \{a_1, a_2, \ldots, a_n, \# \}$.
	Consider the word $S$ which contains all letters $a_i$
	with $\#$ in between and ended with $\#$, first in order from $1$ to $n$,
	then in order $n$ to $1$:
	\begin{equation*}
	S = a_1 \#  a_2 \# a_3 \# \cdots a_{n-1} \# a_n \#
	a_n \# a_{n-1} \cdots a_2 \#  a_1 \# \enspace .
	\end{equation*}	
	Then
	\begin{align*}
	|S|H_0(S)
	&=
	2n\log \frac{4n}{2} + 2n \log \frac{4n}{2n}\\
	&=
	2n \log \left( 2n \cdot 2 \right)\\
	&=
	\frac{|S|}{2} \log |S| \enspace .
	\end{align*}
	
	Every pair occurs twice in $S$.
	We assume that \RePair takes pairs in left-to-right order in case of a tie
	(which is true for standard implementations),
	hence \RePair will produce a starting string
	\begin{equation*}
	S' = X_1 X_2 \ldots X_n X_n X_{n-1} \ldots X_1 \enspace .
	\end{equation*}
	Its entropy is
	\begin{align*}
	|S'|H_0(S')
	&=
	\frac{|S|}{2}\log \frac{|S|/2}{2}\\
	&=
	\frac{|S|}{2}\log |S| - |S| \enspace ,
	\end{align*}	
	i.e.\ the entropy of working string decreased
	only by a lower order term $|S|$ during the execution of \RePair.
	The produced grammar contains $\frac{|S|}{4}$ rules,
	thus by the assumption that grammar of size $g$ takes at least
	$g\log g - \Ocomp(g)$ bits
	the total cost of encoding is at least:
	\begin{equation*}
	\frac{|S|}{2}\log |S| - |S| + \frac{|S|}{4} \log \frac{|S|}{4}- \Ocomp \left(\frac{|S|}{4}\right) = \frac{3|S|}{4} \log |S| - \Ocomp(S) \enspace .
	\end{equation*}

	The same claim holds for entropy encoding of the grammar:
	Let $S_G$ be a string obtained by concatenating 
	the right-hand sides of $(S', G)$.
	We have
	\begin{equation*}
	|S_G| H_0(S_G) = \frac{3|S|}{4} \log |S| - \Ocomp(|S|) \enspace ,
	\end{equation*}
	as each of letters $a_i$ and nonterminals $X_i$ occurs constant
	number of times in $S_G$.
\end{proof}

\paragraph{Weakly non-redundant grammars}
\label{par:weakly_non_red_grammars}
We define a class of \emph{weakly non-redundant} grammars,
which include \RePair and \Greedy.
These grammars have some nice and natural properties,
which we later employ in our analysis of \RePair and \Greedy:
Firstly, the full grammar of a weakly non-redundant grammar
can be obtained by a series of compression steps,
in which we replace only substrings of the working string
(so the rules of the grammars are never altered),
and those substrings occur at least twice in the working string.
In particular, each such compression reduces the length of the starting string
and does not increase the sums of lengths of the right-hand sides
and starting string of the full grammar;
this is showed in Lemma~\ref{lem:weakly_non_redundant_op}.
Secondly, if the entropy of the input string is large,
we can estimate the cost of incremental and entropy codings
of a full grammar in terms of the zeroth-order entropy of the input string,
see Lemma~\ref{lemma15ireplacement} and Lemma~\ref{lemma15replacemententropy}.
These results are of their own interest:
for strings of entropy $H_0(S) \geq \frac{1}{2}\log |S|$
even not-so-reasonable grammar transform (with entropy encoding)
gives an output of size $\frac{3}{2}|S|H_0(S) + \Ocomp(|S|)$.
This is nontrivial, as considered transforms can introduce
up to $|S|/2$ nonterminals, and each nonterminal
in entropy coding seems to require additional $\log|S|$ bits.
Moreover, it seems that we can extend this result
so that for strings of entropy $H_0(S) \geq \frac{1}{\beta}\log |S|$
the output size will be  $((1+\beta)/2 )|S|H_0(S) + \Ocomp(|S|)$.
\begin{definition}(weakly non-redundant grammars)
	\label{def:weakly_non}
	A full grammar $(S',G)$ is \emph{weakly non-redundant}
	if every nonterminal $X$, except the starting symbol,
	occurs at least twice in the derivation tree of $(S',G)$,
	for every nonterminal $X$ we have $\rhs{X} \geq 2$
	and no two different nonterminals have the same expansion.
\end{definition}
For example grammar: $S\rightarrow AA, A\rightarrow Ba, B\rightarrow aa$
is weakly non-redundant, $S\rightarrow AB, A\rightarrow cc, B\rightarrow aa$ is not.

We can obtain any weakly non-redundant grammar
by simple (non-recursive) procedure that at each step 
replaces some substring which occurs at least twice.
\begin{lemma}
	\label{lem:weakly_non_redundant_op}
	Let $S$ be a string and $(S',G)$ be a weakly non-redundant grammar.
	Then $(S',G)$ can be obtained by a series of replacements,
	where each replacement operates only on current starting $S''$,
	replaces at least $2$ occurrences of some substring $w$, where $|w| > 1$
	and $w$ occurs at least 2 times,
	and adds a rule $X_w \rightarrow w$ to the grammar.
\end{lemma}
\begin{proof}
	Take any weakly non-redundant grammar $(S',G)$ generating $S$.
	Consider the following inductive reasoning:
	Take some nonterminal $X \rightarrow ab$,
	where $ab$ are original symbols of $S$,
	i.e. $a,b$ are leaves of the derivation tree.
	Replace each occurrence of $ab$ in $S$ which
	corresponds to nonterminal $X$ in the derivation tree.
	This gives a new string $S''$.
	Now we can remove $X$ from $(S',G)$ and obtain
	a weakly non-redundant grammar for $S''$,
	where lemma holds by the induction hypothesis.
\end{proof}


In the proofs of the following lemmas, the crucial assumption is that $\rhs{X} \geq 2$,
for every nonterminal $X$,
as then each introduced nonterminal shortens the starting string by at least $2$.

By Lemma~\ref{lem:grammar_encodings} the encoding of a full grammar $(S', G)$
uses $|S'|H_0(S') + \nonterminals{G}\log \nonterminals{G}$ bits plus lower order
terms, thus the following Lemma says that we can obtain non-trivial
bound for the encoding of weakly non-redundant grammars when the entropy of string 
$S$ generated by $(S', G)$ is large.

\begin{lemma}
	\label{lemma15ireplacement}
	Let $S$ be a string and $n \in \mathbb{N}$ a number.
	Assume $|S| \leq n$
	and $|S| H_0(S) \geq \frac{|S| \log n}{2} - \gamma$,
	where $\gamma$ can depend on $S$. 
	Let $(S', G)$ be a weakly non-redundant grammar generating $S$.
	Then:
\begin{equation*}
	|S'| H_0(S') + \nonterminals{G} \log n \leq \frac{3}{2} |S| H_0(S) + \Ocomp(|S| + \gamma) \enspace .
\end{equation*}
\end{lemma}

\begin{proof}
	By Lemma~\ref{lem:weakly_non_redundant_op} we can assume that
	grammar is produced by series of replacements on the starting string $S$,
	that is, we never modify right-hand sides of $G$ after adding a rule.
	Because each replacement removes at least two symbols from the working string
	we have
	\begin{equation}
	\label{lem0ord:S_and_S'}
	|S| \geq |S'| + 2\nonterminals{G} \enspace .
	\end{equation}
	Observe that $S'$ induces a parsing of $S$,
	as by grammar properties different symbol expands to a different substring.
	
	Let $L = \Lengths(S')$.
	Then estimating the entropy of parsing from Theorem~\ref{thm:main_estimation},
	and $|L|H_0(L)$ from Lemma~\ref{lem:lengths_entropy}
	(see the Appendix Section~\ref{sec:appendix_c}) yields:
	\begin{equation}
	\label{lem0ord:sparsingestimate}
	|S'| H_0(S') \leq |S| H_0(S) + |L|H_0(L) \leq |S| H_0(S) + \Ocomp(|S|) \enspace .
	\end{equation}
	As $|S'| \leq |S|$ and by assumption $|S| \leq n$ we have
	that $|S|H_0(S) \leq |S| \log n$, and thus:
	\begin{equation}
		\label{lem0ord:boundentropysp}
		|S'|H_0(S') \leq |S'|\log n \enspace .
	\end{equation}
	Now we have:
	\begin{align*}
	|S'|H_0(S')  + \nonterminals{G} \log n
	&= \frac{1}{2}|S'|H_0(S') + \frac{1}{2}|S'|H_0(S')  
	  +\nonterminals{G} \log n \\
	&\leq \frac{1}{2}|S|H_0(S) + \Ocomp(|S|) + \frac{1}{2}|S'|H_0(S')  
	+\nonterminals{G} \log n &\text{by~\eqref{lem0ord:sparsingestimate}} \\
	&\leq \frac{1}{2}|S|H_0(S) + \Ocomp(|S|) + \frac{1}{2}|S'|\log n
	+\nonterminals{G} \log n &\text{by~\eqref{lem0ord:boundentropysp}} \\
	&\leq \frac{1}{2}|S|H_0(S) + \Ocomp(|S|) + \frac{1}{2}|S| \log n &\text{by~\eqref{lem0ord:S_and_S'}}\\
	&\leq \frac{3}{2}|S|H_0(S) + \Ocomp(|S| + \gamma)
	& \text{by assumption \qedhere}
	\end{align*}
\end{proof}

We now show a similar property for the entropy encoding,
its proof is more involved than in the case of Lemma~\ref{lemma15ireplacement},
i.e.\ the incremental encoding.
It uses a similar idea as proof of Theorem~\ref{thm:main_estimation},
that is, we assign some $p$ values to each symbol on the right-hand side of
$(S', G)$ and apply Lemma~\ref{theoremP}.

\begin{lemma}
	\label{lemma15replacemententropy}
	Let $n \in \mathbb N$ be a number and  $S$ be a string, $|S| \leq n$.
	Assume $|S| H_0(S) \geq \frac{|S| \log n}{2} - \gamma$,
	where $\gamma$ can depend on $S$.
	Let $(S', G)$ be a full weakly non-redundant grammar generating $S$.
	Let $S_G$ be a string obtained by concatenation of the right-hand sides of $(S', G)$.
	Then:
	\begin{equation*}
	|S_G| H_0(S_G) \leq \frac{3}{2} |S| H_0(S) + \Ocomp(|S| + \gamma) \enspace .
	\end{equation*}
\end{lemma}

\begin{proof}
	By Lemma~\ref{lem:weakly_non_redundant_op} we can assume that
	grammar is produced by series of replacements on the starting string $S$,
	that is we never modify the right-hand sides of $G$ after adding a rule.
	
	Let $\Gamma$ be the original alphabet of string $S$
	and $S_G$ be a concatenation of strings in $\grhs(S', G)$.
	We show that
	\begin{equation}
	\label{eq:letters_less_in_grammar}
	|S'|_\Gamma + 2 |S_G|_\nont \leq |S|\enspace .
	\end{equation}
	This clearly holds at the beginning.
	Suppose that we replace $k$ copies of $w$ in $S'$ by a new nonterminal $X$
	and add a rule $X \to w$.
	Let $|w|_\Gamma$ and $|w|_\nont$ denote the number of occurrences of letters
	and nonterminals in $w$, then $|w|_\Gamma + |w|_\nont = |w| \geq 2$.
	After the replacement the values
	of~\eqref{eq:letters_less_in_grammar} change as follows:
	\begin{align*}
	(|S'|_\Gamma - k|w|_\Gamma) +  2 (|S_G|_\nont - (k-1)|w|_\nont + k)
	&=
	|S'|_\Gamma + 2|S_G|_\nont - (k-1)(|w|_\Gamma + 2|w|_\nont - 2)  + 2 - |w|_\Gamma\\	
	&\leq
	|S'|_\Gamma + 2|S_G|_\nont - (|w|_\Gamma + 2|w|_\nont - 2)  + 2 - |w|_\Gamma\\
	&=
	|S'|_\Gamma + 2|S_G|_\nont - 2(|w|_\Gamma + |w|_\nont - 2) \\
	&\leq
	|S'|_\Gamma + 2|S_G|_\nont\\
	&\leq
	|S|	 \enspace .
	\end{align*}
	
	Now, observe that for each letter $a \in \Gamma$ it holds that:
	\begin{equation}
	\label{eq:counting_letter_G_and_S}
	|S'|_a + 2\rhs{G}_a \leq |S|_a \enspace .
	\end{equation}
	This is shown by an easy induction:
	clearly, it holds at the beginning.
	When we replace $k \geq 2$ occurrences of a word $w$ with a nonterminal $X$
	and add the rule $X \rightarrow w$
	then $|S'|_a$ drops by $k |w|_a$ while $2\rhs{G}_a$ increases
	by $2|w|_a \leq k|w|_a$.
	
	
	We use Lemma~\ref{theoremP} together with
	\eqref{eq:letters_less_in_grammar} and \eqref{eq:counting_letter_G_and_S}
	to bound the entropy of $S_G$.
	Define $p(a) = \frac{|S|_{a}}{|S|}$ as the empirical probability of letter $a$ in $S$
	and let $p'(a) = \frac{1}{2} p(a)$ for original letters of $S$ and 
	$p'(X) = \frac{1}{2|S|}$ for nonterminal symbols.
	Those values satisfy the condition of Lemma~\ref{theoremP}:
	\begin{align*}
	\sum_{a \in \Gamma \cup \nont(G)} p'(a)
	&=
	\sum_{a \in \Gamma} p'(a)	 + 	\sum_{X \in \nont(G)} p'(X)\\
	&=
	\sum_{a \in \Gamma} \frac{1}{2}p(a) + \sum_{X \in \nont(G)} \frac{1}{2|S|}\\
	&\leq
	\frac{1}{2}\sum_{a \in \Gamma} \frac{|S|_a}{|S|} +	|\nont(G)| \cdot  \frac{1}{2|S|}\\
	&\leq 
	\frac{1}{2} + |S| \cdot \frac{1}{2|S|} & \text{by~\eqref{eq:letters_less_in_grammar}, 
	as $|\nont(G)| \leq |S_G|_\nont$}\\
	&\leq
	1 \enspace .
	\end{align*}
	Thus we can bound the entropy of $S_G$ using Lemma~\ref{theoremP}:
	\allowdisplaybreaks
	\begin{align}
	\label{lem15entmest}
	&|S_G|H_0(S_G)\\
	&\leq
	\sum_a - |S_G|_a \log p'(a)\notag\\
	&=
	\sum_{a \in \Gamma} - |S_G|_a \log p'(a)
	+ \sum_{X \in \nont(G)} - |S_G|_X \log p'(X)\notag\\
	&=
	\sum_{a \in \Gamma} - |S_G|_a \log p(a) + |S_G|_{\Gamma}
	+ \sum_{X \in \nont(G)}|S_G|_X \log |S| + |S_G|_\nont \notag  & 
	\text{as $p'(X) = \frac{1}{2|S|}$}\\
	&=
	\sum_{a \in \Gamma} - \left(\frac{1}{2}|S'|_a + \frac{1}{2}|S'|_a + \rhs{G}_a \right) \log p(a)
	+ |S_G|_\nont \log |S| + |S_G| \notag \\
	&\leq
	\sum_{a \in \Gamma} - \left(\frac{1}{2}|S'|_a + \rhs{G}_a \right) \log p(a)
	+ \sum_{a \in \Gamma} \frac{1}{2} |S'|_a \log |S|
	+ |S_G|_\nont \log |S| + |S_G| \notag & \text{as $p(a) \geq \frac{1}{|S|}$} \\
	&\leq
	\sum_{a \in \Gamma} -\frac{1}{2}|S|_a  \log p(a)
	+ \frac{1}{2}|S'|_\Gamma \log |S|
	+ |S_G|_\nont \log |S| + |S_G| & \text{by~\eqref{eq:counting_letter_G_and_S}}\notag \\
	&=
	\frac{1}{2} |S|H_0(S)
	+ \left(\frac{1}{2}|S'|_\Gamma + |S_G|_\nont \right)\log |S|
	+ |S_G| \notag \\
	&\leq \notag
	\frac{1}{2} |S|H_0(S) + \frac{|S|}{2}\log |S| + |S_G|& \text{by~\eqref{eq:letters_less_in_grammar}}\\
	&\leq \notag
	\frac{1}{2} |S|H_0(S) + |S|H_0(S) + \gamma + |S_G| & \text{by assumption}\\
	&\leq \notag
	\frac{3}{2} |S|H_0(S) + \Ocomp(|S| + \gamma) \enspace \qedhere & \text{as $|S_G|\leq |S|$}
	\end{align}
\end{proof}

\paragraph{Proof of Theorem~\ref{thm:RePairToTheEnd}}
We prove the following fact which we will need in
the proof of Theorem~\ref{thm:RePairToTheEnd} and also later on:
\begin{lemma}
\label{lem:entropy_removing_letters}
Let $S, S'$ be strings such that $S'$ can be obtained from $S$
by reordering and/or removing letters from $S$.
Then $|S|H_0(S) \geq |S'|H_0(S')$.
\end{lemma}
\begin{proof}
Let $\Sigma$ be the alphabet of $S$.
We have:
\begin{equation*}
|S|H_0(S) = \sum_{a\in\Sigma} |S|_a \log \frac{|S|}{|S|_a} \enspace .
\end{equation*}
Clearly reordering letters does not change the above value.

Consider the string $S''$ obtained by deleting a single occurrence of letter $b$ in the string $S$.
Observe that the inequality holds in the trivial case of $S \in b^*$,
as then $|S''|H_0(S'') = |S|H_0(S)=0$,
so in the following we assume that $|S|_b < |S|$.
Then:
\begin{equation*}
|S''|H_0(S'') =  \sum_{a\in\Sigma\setminus\{b\} } |S|_a \log \frac{|S|-1}{|S|_a} 
+ (|S|_b-1)\log\frac{|S|-1}{|S|_b - 1} \enspace ,
\end{equation*}
where we assume that $(|S|_b-1)\log\frac{|S|-1}{|S|_b - 1} = 0$ if $|S|_b = 1$.
As $\log x$ is increasing in range $(0,+\infty)$,
for each $a \neq b$ it holds that
\begin{equation*}
|S|_a \log \frac{|S|-1}{|S|_a} \leq |S|_a \log \frac{|S|}{|S|_a} \enspace .
\end{equation*}
For $b$ consider that
\begin{align*}
(|S|_b-1)\log\frac{|S|-1}{|S|_b - 1}
	&=
\log\left(\left(1 + \frac{|S|-|S|_b}{|S|_b - 1}\right)^{|S|_b - 1}\right)\\
(|S|_b)\log\frac{|S|}{|S|_b}
	&=
\log\left(\left(1 + \frac{|S|-|S|_b}{|S|_b}\right)^{|S|_b}\right)
\end{align*}
and the sequence $f_n = (1 + c/n)^n$ is known to be (strictly)
increasing for every value $c>0$ (as it monotonically tends to $e^{c}$).
Thus
\begin{equation*}
(|S|_b-1)\log\frac{|S|-1}{|S|_b - 1}
	\leq
(|S|_b)\log\frac{|S|}{|S|_b} \enspace ,
\end{equation*}
as those are logarithms of the two consecutive elements of such sequence (for $c = |S|-|S|_b$).
\end{proof}

\begin{proof}[Proof of Theorem~\ref{thm:RePairToTheEnd}]
Fix some $\epsilon > 0$ and consider
the iteration in which the number
of nonterminals in the grammar is equal to $\left\lceil \frac{|S|}{\log^{1 + \epsilon} |S| } \right \rceil$,
call this grammar $G_0$, i.e.\ we have
$\nonterminals{G_0} = \left\lceil \frac{|S|}{\log^{1 + \epsilon} |S| } \right \rceil$.
If no such an iteration exists,
then consider the state after the last iteration, then
$\nonterminals{G_0} < \left\lceil \frac{|S|}{\log^{1 + \epsilon} |S| } \right \rceil$.
Let $S'$ be a working string at the current state,
thus $(S', G_0)$ is a full grammar generating $S$.

By Theorem~\ref{thm:RePairStopped} with $\xi = 2, \beta = 4$
we know that if the size of the grammar is at least $1+ \sqrt{|S|} \cdot \log_\sigma |S|$
then the size of the working string must be at most
$|S'| \leq \frac{14 |S|}{\log_\sigma |S|} + 2\sqrt{|S|} + 1$.
Thus, as
$1+ \sqrt{|S|} \cdot \log_\sigma |S|
 = o\left(\left\lceil \frac{|S|}{\log^{1 + \epsilon} |S| } \right \rceil \right) $
it follows that
$|S'| = \Ocomp\left(\frac{|S|}{\log_\sigma |S|}\right)$.


We make two assumptions: $|S'| > \left\lceil \frac{|S|}{\log^{1 + \epsilon} |S| } \right \rceil$
 and
$\nonterminals{G_0} = \left\lceil \frac{|S|}{\log^{1 + \epsilon} |S| } \right \rceil$
(the latter does not hold if \RePair finished the run with smaller grammar),
as otherwise it is easy to prove the claim:
If $|S'| \leq \left\lceil \frac{|S|}{\log^{1 + \epsilon} |S| } \right \rceil$ then
the size of the final grammar is at most $|S'|/2 + \nonterminals{G_0}$,
as one replacement shortens the working string by at least $2$
thus we can make at most $|S'|/2$ iterations from the current state.
Thus, if at least one of these assumptions is not satisfied
we finish the run of \RePair with 
a grammar of size at most
$\Ocomp\left(\left\lceil \frac{|S|}{\log^{1 + \epsilon} |S| } \right \rceil\right)$,
in such case
by Lemma~\ref{lem:grammar_encodings} we get the bound $|S|H_k(S) + o(|S|\log \sigma)$ for incremental encoding
ant the same bound for entropy encoding by Lemma~\ref{lem:entropy_coding_concatenation}.

We now show that the assumptions: $|S'| > \left\lceil \frac{|S|}{\log^{1 + \epsilon} |S| } \right \rceil$
and
$\nonterminals{G_0} = \left\lceil \frac{|S|}{\log^{1 + \epsilon} |S| } \right \rceil$
imply that the entropy, i.e.\ $|S'|H_0(S')$, is high.
From Lemma~\ref{lem:most_frequent_and_number_of_nonterminals}
we get that each pair occurs at most
$\frac{2|S|}{\nonterminals{G}} \leq 2\log^{1+\epsilon} |S|$ times in $S'$.

Consider the parsing of working string $S'$ into consecutive pairs
(with possibly last letter left-out), denote it by $Y_{S'}^P$.
Let $L_P = \Lengths\left(Y_{S'}^P\right)$, observe that those are all
$2$s except possibly the last, which can be $1$.
By applying Theorem~\ref{thm:main_estimation} with $k = 0$ we get:
\begin{equation}
\label{eq:repair_to_the_end_pairs_lower_bound}
|S'|H_0(S')  + |L_P|H_0(L_P)\geq |Y_{S'}^P|H_0(Y_{S'}^P) \enspace.
\end{equation}
If $L_P \in 2^*$ then clearly $H_0(L_P) = 0$,
otherwise $L_P = 2^{|S'|-1/2}1$
and so its entropy is 
\begin{align*}
|L_p|H_0(L_p) & \leq
\frac{|S'|-1}{2}\log\left(\frac{(|S'|+1)/2}{(|S'|-1)/2}\right) + \log\left(\frac{(|S'|+1)/2}{1}\right)\\
&\leq
\log\left(1 + \frac{1}{(|S'|-1)/2} \right)^{\frac{|S'|-1}{2}} + \log |S'|\\
&\leq
\log \mathrm e + \log |S'| \enspace .
\end{align*}
On the other hand, as each pair occurs at most $2\log^{1 + \epsilon} |S|$ times,
	we can lower bound the entropy by lower bounding
	each summand corresponding to pair in the definition of $H_0$
	by~$\log (|Y_{S'}^P|/\log^{1 + \epsilon})$,
	i.e.\

\begin{align*}
|Y_{S'}^P| H_0(Y_{S'}^P)
&\geq
|Y_{S'}^P| \log\left( \frac{|S'|-1}{2}/(2\log^{1+\epsilon} |S|) \right)\\
&\geq
\frac{|S'|-1}{2}\log\left( \frac{|S|}{4 \log^{2+2\epsilon} |S|}\right)\\
\notag
&=
\frac{|S'|-1}{2}\log |S| - \frac{|S'|-1}{2}\log (4\log^{2+2\epsilon} |S|)\\
\notag
&\geq
\frac{|S'|-1}{2}\log |S| - (|S'|-1)(1+ (1+\epsilon)\log \log |S|)\\
\label{eq:parsingtopairsentropy}
&\geq
\frac{|S'|}{2}\log |S| - \gamma' \enspace,
\end{align*}
where $\gamma' =  (|S'|-1)(1+ (1+\epsilon)\log \log |S|) + \frac{\log |S| }{2} = o(|S| \log \sigma)$.
Plugging those two estimations into~\eqref{eq:repair_to_the_end_pairs_lower_bound} yields:
\begin{equation}
\label{eq:RePairS'gamma}
|S'|H_0(S')
\geq
\frac{|S'|}{2}\log |S| - \gamma \enspace ,
\end{equation}
where
$\gamma = \gamma' +
\log \mathrm e + \log |S'| = o(|S| \log \sigma)$.

On the other hand, by Theorem~\ref{thm:main_estimation},
\begin{equation}
\label{eq:s'entropy}
|S'|H_0(S') \leq |S|H_k(S) + o(|S| \log \sigma) \enspace .
\end{equation}

Recall that $\nonterminals{G_0} = \left\lceil \frac{|S|}{\log^{1 + \epsilon} |S| } \right \rceil$,
$|S'| = \Ocomp\left(\frac{|S|}{\log_\sigma |S|}\right)$
and so the encoding at this point takes $|S|H_k(S) + o(|S|\log \sigma)$ bits
for both incremental and entropy encoding
by Lemma~\ref{lem:grammar_encodings} and
Lemma~\ref{lem:entropy_coding_concatenation} respectively.
We now use the properties of newly defined
weakly non-redundant grammars~(see~\ref{par:weakly_non_red_grammars}),
which essentially state that if the entropy is $H_0(S') \gtrapprox \frac{ \log |S|}{2}$ 
then encoding size can grow to at most $\frac{3}{2}|S'|H_0(S')$,
assuming that adding a non terminal costs $\log |S|$ bits
(Lemma~\ref{lemma15ireplacement} and Lemma~\ref{lemma15replacemententropy}).
We argue that this is the case for both incremental and entropy coding.

We first prove the bound for incremental encoding.
Let $S''$ be a string returned by \RePair at the end of its runtime
and $G'$ be a grammar at this point.
Observe that $S''$ induces a parsing of $S$,
as $S'$ induces a parsing of $S$ and $S''$ induces a parsing of $S'$.
Let $i$ be such
that the returned grammar $G'$ has $\nonterminals{G_0}+i$
rules,
i.e.\ there were $i$ compression steps between $G_0$ and the end.
Thus $(S'', G')$ is a full grammar which generates $S$,
thus by Lemma~\ref{lem:grammar_encodings} the size of \cnfenc encoding is at most 
\begin{align}
\label{eq:repair_to_the_end_estimation_final_grammar}
\notag &|S''|H_0(S'') + \nonterminals{G'}\log(\sigma + \nonterminals{G'})
+ \Ocomp(|S''| + \nonterminals{G'}) \enspace &\\ 
\leq &|S''|H_0(S'') + \nonterminals{G'}\log(2|S|) + \Ocomp(|S''| + \nonterminals{G'}) 
\notag & \text{as $\sigma,\nonterminals{G'} \leq |S|$}  \\
= &|S''|H_0(S'') + \nonterminals{G'}\log|S| + \Ocomp(|S''| + \nonterminals{G'}) \notag \\
= &|S''|H_0(S'') + (\nonterminals{G_0} + i)\log |S| + \Ocomp(|S''| + \nonterminals{G'})
\enspace .
\end{align}

We can look at the last $i$ iterations
as we start with string $S'$ and end up with $S''$ and some grammar $G' \setminus G_0$,
i.e.~grammar $G'$ without nonterminals from $G_0$.
As this grammar meets the conditions of a weakly non-redundant grammar,
and $S'$ satisfies~\eqref{eq:RePairS'gamma}, i.e.\
$|S'|H_0(S') \geq\frac{|S'|}{2}\log |S| - \gamma$,
we can apply Lemma~\ref{lemma15ireplacement} with $n=|S|$
(recall that this lemma requires $|S'|\leq n$):
\begin{equation}
\label{eq:repair_to_the_end_weakly_bound}
i \log |S| + |S''| H_0(S'') \leq \frac{3}{2} |S'|H_0(S') + \Ocomp(|S'| + \gamma )  \enspace.
\end{equation}
Note that
$\nonterminals{G'} = \nonterminals{G_0} + i = \Ocomp\left(\frac{|S|}{\log_\sigma|S|}\right)$,
as $\nonterminals{G_0} = \left\lceil \frac{|S|}{\log^{1 + \epsilon} |S| } \right \rceil$
and $i < |S'| = \Ocomp\left(\frac{|S|}{\log_\sigma |S|}\right)$.
Combining this fact with~\eqref{eq:repair_to_the_end_estimation_final_grammar}
and~\eqref{eq:repair_to_the_end_weakly_bound} we obtain that the
encoding of grammar $(S'', G')$ returned by \RePair is at most:
\begin{align*}
	&
|S''|H_0(S'') + (\nonterminals{G_0} + i)\log |S| + \Ocomp(|S''| + \nonterminals{G'}) \\
	\leq
& \frac{3}{2} |S'|H_0(S') + \Ocomp(|S'| + \gamma ) + 
\underbrace{\nonterminals{G_0}\log |S|}_{=\left\lceil \frac{|S|}{\log^{1 + \epsilon} |S| } \right \rceil \cdot \log |S|} + \Ocomp(|S''| + \nonterminals{G'}) \enspace &\\
	\leq
& \frac{3}{2} |S'|H_0(S') +
\Ocomp\left(\frac{|S|}{\log^{\epsilon} |S|}\right) + \Ocomp(|S'| + \gamma + |S''| + \nonterminals{G'}) \\
	\leq
& \frac{3}{2} |S'|H_0(S') + o(|S|\log \sigma) \\
	\leq
& \frac{3}{2} |S|H_k(S) + o(|S|\log \sigma) \enspace , & \text{by~\eqref{eq:s'entropy}}
\end{align*}
which ends the proof for the \cnfenc encoding.

We now move to the case of entropy coding.
Let $S_G$ be the concatenation of the right-hand sides of $(S', G_0)$.
Using Lemma~\ref{lem:entropy_coding_concatenation}
we estimate the entropy of $S_G$ as :
\begin{align}
\label{rpthm15:entrhk15}
|S_G| H_0(S_G)
&\leq
|S|H_k(S) +  \nonterminals{G_0} \log |S| +  o(|S| \log \sigma) \enspace .
\end{align}
On the other hand we can show that $|S_G|H_0(S_G)$ is large,
(note that we use the fact that $S'$ can be obtained from $S_G$ by removing/reordering letters) i.e.:
\begin{align}
\label{lem:repair_to_the_end_large_entropy_coding}
|S_G| H_0(S_G) &\geq \notag
|S'| H_0(S')  & \text{by Lemma~\ref{lem:entropy_removing_letters}}\\
&\geq \notag
\frac{|S'|}{2} \log |S| - \gamma &\text{from~\eqref{eq:RePairS'gamma}}\\
&= \notag
\frac{|S_G|}{2} \log |S| - \gamma - 2\nonterminals{G_0}\log |S| &
\text{as $|S_G|=|S'|+2\nonterminals{G_0}$}\\
&\geq
\frac{|S_G|}{2} \log |S| - \gamma'' \enspace,
\end{align}
where $\gamma'' = \gamma + 2\nonterminals{G_0} \log |S| = o(|S|\log \sigma)$.

Define $S'_G$ as the concatenation of the right-hand sides
of a full grammar $(S'', G')$, where $G'$
is a grammar at the end of the runtime of \RePair
and $S''$ is the working string returned by \RePair.
Again, we can look at the last iterations of \RePair like we would start
some compressor with input string $S_G$ and end up with $(S'', G' \setminus G_0)$,
as we can view the remaining iterations like they would do replacements
on $S_G$.
As $(S'', G' \setminus G_0)$ meets the conditions of a weakly non-redundant grammar
and $S_G$ has large entropy~\eqref{lem:repair_to_the_end_large_entropy_coding}
applying Lemma~\ref{lemma15replacemententropy} with $n=|S|$ to $(S'', G' \setminus G_0)$ gives us:
\begin{equation}
\label{rpthm15:entrgrammar}
|S'_G|H_0(S'_G) \leq \frac{3}{2}|S_G|H_0(S_G) + \Ocomp(|S_G| + \gamma'') \enspace .
\end{equation}
Combining~\eqref{rpthm15:entrgrammar} with~\eqref{rpthm15:entrhk15} gives us:
\begin{align*}
|S'_G|H_0(S'_G) \leq & \frac{3}{2}\left( |S|H_k(S) +  \nonterminals{G_0} \log |S| +  o(|S| \log \sigma) \right) + \Ocomp(|S_G| + \gamma'') \\
	\leq & \frac{3}{2}|S|H_k(S) + o(|S|\log\sigma) \enspace ,
\end{align*}
as $\nonterminals{G_0} = \left\lceil \frac{|S|}{\log^{1 + \epsilon} |S| } \right \rceil$
and $|S_G| = |S'| + 2\nonterminals{G_0} =  \Ocomp\left(\frac{|S|}{\log_\sigma |S|}\right)$.

The bound on the entropy coding comes from the fact that by Lemma~\ref{lem:huffman_entropy_coder} entropy coder uses $|S'_G|H_0(S'_G) + \Ocomp(|S'_G|)$ bits to encode the string $S'_G$,
thus the bound $\frac{3}{2}|S|H_k(S) + o(|S|\log\sigma)$ holds,
as $|S'_G| = \Ocomp\left(\frac{|S|}{\log_\sigma |S|}\right)$.
\end{proof}

\subsection{Proofs for Section~\ref{subsec:irreducible} (Irreducible grammars and their properties)}
\paragraph{Irreducible grammars --- lower bounds}

\begin{proof}[Proof of Lemma~\ref{lem:irreducibleDecompression}]
For an integer $i$ consider the grammar
which for each possible binary string $wa$ of length $2 < |wa| \leq i$
has a rule  $X_{wa} \rightarrow X_w a$
and $X_w \to w$ for $w$ of length $2$
and the starting rule $S$ contains each $X_w$ for $|w| = i$ twice,
in lexicographic order on indices.
For instance for $i = 3$:\\
\begin{align*}
S&\rightarrow X_{000} X_{000} X_{001} X_{001} \cdots X_{111} X_{111};\\
X_{000}&\rightarrow X_{00} 0 ; X_{001}\rightarrow X_{00} 1; \ldots ; \\
X_{00}&\rightarrow 00; X_{01} \rightarrow 01;\ldots.
\end{align*}
This grammar is irreducible, has size $\Theta(2^i)$
and generates a string $S_i$ of size $\Theta(i 2^i)$,
i.e.\ the grammar is of size $\Theta(|S_i|/\log_\sigma |S_i|)$.

We show that decompressing any set of nonterminals such that
$m$ nonterminals remain
yields a grammar with a starting rule of size at least $i2^{i+1} - 2^i \log m$.
After the decompression the starting rule $S'$ is of the following form:
each $X_w$ is either untouched, fully expanded to string $w$ or replaced with $X_{w'}B$,
where $B$ is a binary string satisfying $w=w'B$.
For the uniformity of presentation we will consider the last two cases together,
introducing dummy nonterminal $X_\epsilon \to \epsilon$.
Then, there are at most $m+1$ different nonterminals (including $X_\epsilon$)
on the right-hand sides of $G$.
Let $n_0, n_1, n_2, \ldots, n_m$ denote the number of occurrences of such nonterminals in
starting string $S'$;
in principle, it could be that $n_i = 0$, when some expanded nonterminal did not occur in $S'$.
Observe that
\begin{equation}
\label{eq:number_of_nonterminals}
\sum_{j=1}^m n_j = 2^{i+1} \enspace.
\end{equation}
At the same time, when a nonterminal occurs $n_j$ times,
then it appeared for expanding at least $n_j/2$ different nonterminals
and so it introduced at least $n_j \log (n_j/2)$ new letters,
as those $n_j/2$ nonterminals need to have a common prefix
(when $n_j = 0$ then we set $n_j \log(n_j/2) = 0$,
which is fine by the continuity of the function $x \log x$).
Thus we extended the starting rule by at least
\begin{equation}
\label{eq:number_of_new_letters}
\sum_{j=1}^m n_j \log(n_j/2)
\end{equation}
new letters.

The function $x \log (x/2)$ is convex and so (by Jensen inequality)
the value of \eqref{eq:number_of_new_letters}
is minimized, when all $n_j$ are equal, that is
\begin{align*}
\sum_{j=1}^m n_j \log(n_j/2)
	&\geq
\left(\sum_{j=1}^m n_j\right) \log \left(\frac{\sum_{j=1}^m n_j}{2m}\right) &\text{$x \log x$ is concave}\\
	&=
2^{i+1} \log \left(\frac{2^i}{m}\right) &\text{by~\eqref{eq:number_of_nonterminals}}\\
	&=
i 2^{i+1} - 2^i \log m\enspace ,
\end{align*}
as claimed.

Lastly, it is easy to see that if $m = \Ocomp(|S_i|^c) = \Ocomp((i2^i)^c)$ for some $0 < c < 1$
then for some constant $\alpha$
\begin{align*}
i 2^{i+1} - 2^i \log m
	&\geq
i2^{i+1} - 2^i \log \left(\alpha (i 2^i)^c \right)\\
	&=
i2^{i+1}(1 - c) - 2^i (c \log i + \log \alpha)
\end{align*}
which is linear (in $|S_i| = i 2^{i+1}$).
\end{proof}

\begin{proof}[Proof of Lemma~\ref{lem:20irreduciblelowerbound}]
We consider the family of grammars from Lemma~\ref{lem:irreducibleDecompression}
and strings generated by them.
For a parameter $i$ each such string $S_i$ have length $2^{i+1}\cdot i$,
its zeroth-order entropy is clearly $1$, thus $|S_i|H_0(S_i) = |S_i|$.
Now, in the entropy coding we take the concatenation of productions' right-hand sides $S_G$.
Each nonterminal $X_w$ occurs twice in $S_G$,
and the total number of nonterminals $X_w$ is $2^i + 2^{i-1} + \ldots 2^{2} = 2^{i+1} - 4$,
as for each binary word of length $j$, $2\leq j \leq i$ there is a nonterminal representing it.
Hence, we can lower bound the entropy of $S_G$:
\begin{align*}
|S_G|H_0(S_G) &\geq 2\cdot \left(2^{i+1} - 4\right)\log \frac{|S_G|}{2}\\ 
&\geq 2^{i+2} \log 2^{i-\Ocomp(1)} - \Ocomp(\log 2^{i-\Ocomp(1)}) \\
&\geq 2^{i+2} i - \Ocomp(2^{i+2}) \\
&= 2|S_i| - o(|S_i|\log \sigma) \enspace .
\end{align*}
As $|S_i|$ can be arbitrarily large, the claim holds.
\end{proof}

\paragraph{Irreducible grammars --- upper bounds}
We state some properties of irreducible grammars which will be used in
proof of Theorem~\ref{thm:irreducible_entropy_concatenation},
proof of Theorem~\ref{thm:irreducible_naive},
and to obtain bounds for \Greedy.

The following lemma states a well-known property of irreducible grammars,
which in facts holds even if only \IGref{2} is satisfied.
\begin{lemma}[{\cite[Lemma~4]{SmallestGrammar}}]
	\label{expansions_two_times}
	Let $(S', G)$ be a full grammar generating $S$ in which
	each nonterminal (except for the starting symbol)
	occurs at least twice on the right-hand side, 
	that is it satisfies condition~\IGref{2}.
	Then the sum of expansions of the right sides is at most $2|S|$.
\end{lemma}

First, we show that irreducible grammars are necessarily small, similar to what was shown in~\cite{kieffer2000grammar},
yet we will use more general lemma, which is useful in discussion of the~\Greedy algorithm.

\begin{lemma}
\label{lemma_grammar_size_irreducible}
Let $(S', G)$ be a full grammar satisfying conditions~\IGreftwo{1}{2},
let $\xi, \beta > 0$ be constants
and let $S$ be the string generated by grammar $(S', G)$.
If $\rhs{S', G} \geq \frac{ (2\xi + 10\beta) |S|}{\log_\sigma |S|} + 8|S|^\frac{2}{\beta} + 2$
then there is a digram occurring at least $\frac{\xi |S|^{1-\frac{2}{\beta}}}{\log_\sigma |S|} + 4$
times in the right-hand sides of~$(S',G)$.
\end{lemma}
\begin{corollary}
\label{cor:irreducible_is_small}
For any grammar full $(S',G)$ satisfying conditions \IGrefall
it holds that $\rhs{S',G} \leq \frac{ 42|S|}{\log_\sigma |S|} + 8\sqrt{|S|} + 2
= \Ocomp \left( \frac{|S|}{\log_\sigma |S|} \right)$,
where $S$ is a string generated by $(S', G)$.
\end{corollary}

\begin{proof}[Proof of Lemma~\ref{lemma_grammar_size_irreducible}]
The proof is similar to the proof of Theorem~\ref{thm:RePairStopped}.
Consider symbols in strings in $\grhs{(S', G)}$,
we distinguish between \emph{long} and \emph{short} symbols,
depending on the length of their expansion:
a symbol $X$ is long if it is a nonterminal,
and $|\exp(X)| \geq \frac{\log_\sigma |S|}{\beta}-1$,
otherwise it is short.
By $c_l$ denote number of occurrences of long symbols on the right-hand side of 
full grammar $(S', G)$,
and by $c_s$ number of short occurrences,
let also $c = \rhs{S',G}$, then $c = c_l + c_s$.
From Lemma~\ref{expansions_two_times} the sum of expansions' lengths
is at most $2|S|$ and so
	\begin{equation*}
c_l \cdot \left(\frac{\log_\sigma |S|}{\beta} - 1\right)
\leq 2|S| \quad \implies \quad 
c_l
\leq
\frac{2\beta|S|}{\log_\sigma |S| - \beta}  \enspace .
\end{equation*}
Assume $\beta \leq \frac{1}{5}\log_\sigma |S|$,
as otherwise $c > 2|S|$ thus the Lemma trivially holds.
Observe that on the right-hand sides there are exactly
$c-1-\nonterminals{G}$ digrams.
As each production is of size at least $2$ (thus $2\nonterminals{G} \leq c$)
and any long symbol can partake in $2$ digrams,
the number of digrams which consists of two short symbols is at least:
\begin{align}
c-1-\nonterminals{G} - 2\cdot c_l
&\notag\geq c - 1 - \frac{c}{2} - 2\cdot c_l 
	&\text{as $2\nonterminals{G} \leq c$} \\
&\notag= \frac{c}{2} -1 - 2\cdot c_l \\
&\notag\geq
\frac{1}{2}\left(\frac{(2\xi+10\beta)\cdot |S|}{\log_\sigma |S|} + 8|S|^\frac{2}{\beta} + 2 \right)- 1 - \frac{4\beta \cdot |S|}{\log_\sigma|S| - \beta}\\
&\notag\geq
\frac{(\xi+5\beta)\cdot |S|}{\log_\sigma |S|} + 4|S|^\frac{2}{\beta} - 
\frac{4\beta \cdot |S|}{\log_\sigma|S| - \frac{1}{5}\log_\sigma |S| }
	&\text{as $\beta \leq \frac{1}{5}\log_\sigma |S|$}\\
&\notag\geq
\frac{(\xi+5\beta)\cdot |S|}{\log_\sigma |S|} + 4|S|^\frac{2}{\beta} - 
\frac{5\beta \cdot |S|}{\log_\sigma|S| }\\
&\notag\geq
\frac{\xi |S|}{\log_\sigma |S|} + 4|S|^\frac{2}{\beta} \enspace .
\end{align}
On the other hand, by \IGref{1} distinct nonterminals 
corresponds to distinct substrings of $S$,
thus the number of different short symbols is at most:
\begin{equation*}
\sum_{ i = 1}^{\left\lfloor \frac{\log_\sigma |S|}{\beta} - 1\right\rfloor }  \sigma^i
< \sigma^\frac{\log_\sigma |S|}{\beta} = |S|^\frac{1}{\beta} \enspace,
\end{equation*}
thus the number of different digrams which consists of two short symbols is 
$|S|^\frac{1}{\beta} \cdot |S|^\frac{1}{\beta} = |S|^\frac{2}{\beta}$.
Then there must exist a digram occurring at least
\begin{equation*}
\frac{\frac{\xi |S|}{\log_\sigma |S|} + 4|S|^\frac{2}{\beta}}{|S|^\frac{2}{\beta}}
\geq \frac{\xi |S|^{1-\frac{2}{\beta}}}{\log_\sigma|S|} + 4
\end{equation*}
times.
\end{proof}

\begin{proof}[Proof of Theorem~\ref{thm:irreducible_entropy_concatenation}]
Let $(S', G)$ be an irreducible full grammar generating $S$
and $S_G$ be a concatenation of strings in $\grhs{(S', G)}$.

As in the proof of Lemma~\ref{lem:entropy_coding_concatenation}
from $(S', G)$ we will construct a parsing of $S$ by expanding nonterminals:
set $S''$ to be the starting string $S'$
and while possible, take $X$ that occurs in $S''$
and was not chosen before then replace a single occurrence
of $X$ in $S''$ with $\grhs(X)$.
Since every nonterminal is used in the production of $S$,
we end up with $S''$ in which every nonterminal
was expanded once.

Let $S_{\mathcal{N}}$ be a string in which every nonterminal of $G$ occurs once.
Observe that $S''S_{\mathcal{N}}$ has the same count of every symbol
as $\grhs(S',G)$:
each replacement performed on $S''$ removes a single occurrence
of on nonterminal and each nonterminal is processed once.
Thus it is enough to estimate $|S''S_{\mathcal{N}}|H_0(S''S_{\mathcal{N}}) = |S_G|H_0(S_G)$,
as the order of symbols does not affect zeroth-order entropy.

We first estimate $|S''|H_0(S'')$.
To this end observe that $S''$ induces a parsing of $S$
and $|S''| < |S''S_{\mathcal{N}}| = \rhs{S',G}$ 
and as $(S',G)$ is irreducible,
by Corollary~\ref{cor:irreducible_is_small} it is small,
i.e.\ $\rhs{S',G} = \Ocomp\left(\frac{|S|}{\log_\sigma |S|}\right)$.
Hence also $|S''| = \Ocomp\left(\frac{|S|}{\log_\sigma |S|}\right)$ and by Theorem~\ref{thm:main_estimation} we have:
\begin{equation}
\label{eq:irreducible_entropy_concatenation_S''_entropy}
|S''|H_0(S'') \leq |S|H_k(S) + o(|S|\log \sigma) \enspace .
\end{equation}

Each nonterminal is present in $S''$:
by \IGref{2} each nonterminal had at least two occurrences
in the right-hand sides of $(S', G)$,
and only one occurrence was removed in the creation of $S''$.
Now, consider the string $S''S''$,
by the above reasoning we can reorder characters of $S''S''$
and remove some of them to obtain the string $S'' S_{\mathcal{N}}$.
Thus, as by Lemma~\ref{lem:entropy_removing_letters}
removing and reordering letters from string $w$ can only decrease $|w|H_0(w)$ we have:
\begin{align*}
|S_G|H_0(S_G)
&=
(|S''S_{\mathcal{N}}|) H_0(S'' S_{\mathcal{N}}) &\text{same symbol count}\\
&\leq
(|S''S''|)H_0(S''S'') &\text{greater symbol count}\\
&=
2 |S''| H_0(S'') \\
&\leq
2|S|H_k(S) + o(|S|\log\sigma) &\text{by \eqref{eq:irreducible_entropy_concatenation_S''_entropy}} \enspace,
\end{align*}
and the same bound holds for entropy coding,
which by Lemma~\ref{lem:huffman_entropy_coder} consumes $|S_G|H_0(S_G) + \Ocomp(|S_G|)$ bits,
as $|S_G| = \frac{|S|}{\log_\sigma |S|} = \Ocomp(|S|\log \sigma)$.
\end{proof}

\begin{proof}[Proof of Theorem~\ref{thm:irreducible_naive}]
Let $(S', G)$ be an irreducible full grammar generating $S$
and $S_G$ be a concatenation of strings in $\grhs{(S', G)}$.
Since $(S', G)$ is irreducible by Corollary~\ref{cor:irreducible_is_small}
$|S_G| = \Ocomp\left(\frac{|S|}{\log_\sigma |S|}\right)$.
Moreover, by \IGref{3} no digram occurs twice (without overlaps)
in the right-hand sides of $(S', G)$.
Therefore there exist $ \frac{1}{3} |S_G|$
non-overlapping pairs in $S_G$ that occur exactly once:
we can pair letters in each production, except the last letter if production's length is odd.
Let $Y_{S_G}$ be a parsing of $S_G$ into those pair
and single symbols.
Then
\begin{align*}
|Y_{S_G}|H_0(Y_{S_G})
&\geq 
\frac{1}{3} |S_G| \log \frac{1}{3} |S_G|\\
&=
\frac{1}{3} |S_G| \log |S_G| - \Ocomp(|S_G|)\\
&=
\frac{1}{3} |S_G| \log |S_G| - o(|S| \log \sigma)
\enspace .
\intertext{	On the other hand,
	applying Theorem~\ref{thm:main_estimation}
	for $k=0$
	and using the estimation on $|L|H_0(L) = \Ocomp(|S_G|)$ from
	Lemma~\ref{lem:lengths_entropy}
	(see the Appendix Section~\ref{sec:appendix_c})
	we get:
}	|S_G|H_0(S_G)
&\geq
|Y_{S_G}|H_0(Y_{S_G}) - o(|S|\log \sigma)
\intertext{and thus}
|S_G|H_0(S_G)
&\geq
\frac{1}{3} |S_G| \log |S_G| - o(|S|\log \sigma) \enspace.
\intertext{Now, from Theorem~\ref{thm:irreducible_entropy_concatenation}:}
2|S| H_k(S) + o(|S| \log \sigma )
&\geq
|S_G|H_0(S_G)\\
&\geq
\frac{1}{3} |S_G| \log |S_G| - o(|S|\log \sigma)	
\intertext{and so}
|S_G| \log |S_G| &\leq 6|S|H_k(S) + o(|S|\log \sigma) \enspace .
\end{align*}
The naive coding of $S_G$ uses at most
$|S_G| \log |S_G| + |S_G| = |S_G| \log |S_G| + o(|S| \log \sigma)$ bits,
note that $S_G$ includes all letters from the original alphabet.
\end{proof}

\subsection{Proofs for Section~\ref{sec:greedy} (\Greedy) }
We state properties of \Greedy that are similar to those of \RePair.

\begin{lemma}[cf.~Lemma~\ref{lem:Repairfrequencedoesnotdecrease}]
	\label{greedy_freq_does_not_increase}
	Frequency of the most frequent pair in the working string and grammar right-hand sides
	does not increase during the execution of \Greedy.
\end{lemma}
\begin{proof}
	Assume that after some iteration which introduced
	symbol $X \leftarrow w$
	frequency of some pair $AB$ increased.
	Then the only possibility is that either $A = X$ or $B=X$,
	as otherwise the frequency of $AB$ cannot increase.
	By symmetry let us assume that $A = X$,
	for the moment assume also that $A \neq B$.
	Let $A_k$ be the last symbol of $w$;
	then $A_k B$ occurred at least as many times as $AB$
	in the working string and grammar before replacement.
	The case with $A = B$ is shown in the same way;
	the case with $X = B$ is symmetric.
\end{proof}

\begin{lemma}[cf.~Lemma~\ref{lem:most_frequent_and_number_of_nonterminals}]
	\label{greedy_grammar_size_pair}
	Let $z$ be a frequency of the most frequent pair in the working string and grammar right-hand sides at some point in execution of \Greedy.
	Then at this point the number of nonterminals in the grammar is at most $\nonterminals{G} \leq \frac{2|S|}{z-4}$.
\end{lemma}
\begin{proof}
	Replacing $f_w$ occurrences of a substring $w$ of length $|w|$
	decreases the $\rhs{S',G}$ by
	\begin{equation*}
	\underbrace{f_w(|w|-1)}_{\text{$w$ is replaced with one symbol}} - \underbrace{|w|}_{\text{rule of length $|w|$ added}} = 
	(f_w -1)(|w| - 1) - 1 \enspace .
	\end{equation*}
	In particular, replacing $f$ occurrences of a pair of symbols
	shortens the grammar by $f-2$ symbols.
	Now, by Lemma~\ref{greedy_freq_does_not_increase}
	in each previous iteration, there was a pair with frequency
	at least $z$.
	As \Greedy replaces the string which shortens
	the working string and right-hand sides of rules
	by the maximal possible value,
	in each previous iteration $|S'| + \rhs{G}$
	was decreased by at least $z/2 - 2 = (z-4)/2$
	(note that occurrences of pairs may overlap)
	so the maximal number of previous iterations is $\frac{2|S|}{z-4}$
	and this is also the bound on the number of added nonterminals.
\end{proof}

\begin{lemma}[\cite{SmallestGrammar}]
\label{lem:greedy_properties}
Greedy produces irreducible grammar.
At each iteration of \Greedy the 
conditions \IGreftwo{1}{2} are satisfied for the grammar $(S, G)$,
i.e.\ no two different nonterminals have the same expansions
and each nonterminal is used a least twice.
\end{lemma}
\begin{proof}
Both claims were proved in~\cite{SmallestGrammar},
the first one is explicitly stated,
the second claim follows from \Greedy being
a global algorithm (see~\cite[Section D]{SmallestGrammar}).
\end{proof}

\begin{proof}[Proof of Theorem~\ref{GreedyStopped}]
Let $(S', G)$ be a full grammar generating $S$ at some point in the execution of \Greedy.
As \Greedy produces irreducible grammar
and at each iterations conditions \IGreftwo{1}{2} are satisfied
by Lemma~\ref{lem:greedy_properties},
we can apply Lemma~\ref{lemma_grammar_size_irreducible}:
if $\rhs{S', G} \geq \frac{ (2\xi + 10\beta) |S|}{\log_\sigma |S|} + 8|S|^\frac{2}{\beta} + 2$
then there exists a digram occurring at least
$\frac{\xi |S|^{1-\frac{2}{\beta}}}{\log_\sigma |S|} + 4$ times.
Lemma~\ref{greedy_grammar_size_pair} applied at this point in the execution
gives the bound on the number of nonterminals:
\begin{align*}
\nonterminals{G}
&\leq
\frac{2|S|}{ \left(\frac{\xi |S|^{1-\frac{2}{\beta}}}{\log_\sigma |S|} + 4\right) - 4}\\
&\leq
\frac{2}{\xi}|S|^\frac{2}{\beta}\log_\sigma |S| \enspace .
\end{align*}
As some digram occurs $\frac{\xi |S|^{1-\frac{2}{\beta}}}{\log_\sigma |S|} + 4 > 2$ times,
then there must be the next iteration, and thus there must exist iteration when
$\rhs{S', G} < \frac{ (2\xi + 10\beta) |S|}{\log_\sigma |S|} + 8|S|^\frac{2}{\beta} + 2$
for the first time.
The size of the grammar at this moment is at most
$\frac{2}{\xi}|S|^\frac{2}{\beta}\log_\sigma |S| + 1$.

As at this moment $\rhs{S', G} = \Ocomp\left( \frac{|S|}{\log_\sigma |S|} \right)$
and $\nonterminals{G}\log |S| = o(|S|)$ for $\beta > 2$,
by Lemma~\ref{lem:entropy_coding_concatenation}
the entropy coding uses at most $|S|H_k(S) + o(|S|\log\sigma)$ bits.
\end{proof}

\begin{proof}[Proof of Theorem~\ref{thm:greedy_naive}]
Let $(S'', G')$ be a full grammar produced by \Greedy at the end of the execution
and let $S'_G$ be the concatenation of the right-hand sides of $(S'', G')$.

Similarly as in the proof of Theorem~\ref{thm:RePairToTheEnd},
we will lower bound the entropy of the grammar $(S', G)$,
which is a grammar at a certain iteration,
and apply the results for weakly non-redundant grammars,
see Section~\ref{par:weakly_non_red_grammars}.

Fix some $\epsilon > 0$ and consider the iteration in which
$\nonterminals{G} =
\left\lceil
\frac{|S|}{\log^{1 + \epsilon} |S|}
\right\rceil$;
if there is no such an iteration then consider the grammar produced at the end.
Let $(S', G)$ be a full grammar at this iteration,
and let $S_G$ be a concatenation of the right-hand sides of $(S', G)$.

By Lemma~\ref{greedy_grammar_size_pair} 
if the size of the grammar is at least $2\sqrt{|S|}\log_\sigma |S| + 1$,
then every pair occurs less than $\sqrt{|S|}{\log_\sigma|S|} + 4$ times.
Moreover,  as grammar produced by \Greedy at each iteration satisfies~\IGreftwo{1}{2},
by Lemma~\ref{lemma_grammar_size_irreducible} with $\xi = 1, \beta = 4$
if every pair occurs
less than $\sqrt{|S|}{\log_\sigma|S|} + 4$ times then
$\rhs{S', G} < \frac{ 42 |S|}{\log_\sigma |S|} + 8\sqrt{|S|} + 2$.
Thus if the grammar is of size at least $2\sqrt{|S|}\log_\sigma |S| + 1$
then $|S_G| < \frac{ 42 |S|}{\log_\sigma |S|} + 8\sqrt{|S|} + 2$.
As
$2\sqrt{|S|} \cdot \log_\sigma |S| + 1
= o\left(\left\lceil \frac{|S|}{\log^{1 + \epsilon} |S| } \right \rceil \right) $
it follows that
$|S_G| = \Ocomp\left(\frac{|S|}{\log_\sigma |S|}\right)$.
Moreover, by Lemma~\ref{lem:entropy_coding_concatenation} we have that:
\begin{align}
\label{eq:greedy_sg_upper_bound}
|S_G|H_0(S_G) &\leq |S|H_k(S) + \nonterminals{G} \log |S| + o(|S|\log \sigma) \notag \\
&\leq |S|H_k(S) + o(|S|\log \sigma)
\end{align}

We assume that $|S_G| > \frac{|S|}{\log^{1 + \epsilon} |S|}$,
otherwise, as
$\rhs{S', G}$ never increases, it holds that $|S'_G| \leq |S_G|$,
and then $|S'_G|H_0(S'_G) \leq |S'_G| \log |S'_G|  \leq |S_G| \log |S_G|= o(|S|)$
thus both the naive and entropy coding takes $o(S)$ bits.

Now, we have that either $\nonterminals{G} =
\left\lceil
\frac{|S|}{\log^{1 + \epsilon} |S|} \right\rceil$,
or
Greedy finished the execution with smaller grammar, i.e.
$\nonterminals{G} < \left\lceil \frac{|S|}{\log^{1 + \epsilon} |S|} \right\rceil$.
If $\nonterminals{G} = \left\lceil \frac{|S|}{\log^{1 + \epsilon} |S|} \right\rceil$
then by Lemma~\ref{greedy_grammar_size_pair} no digram occurs more
than $\frac{2|S|}{\nonterminals{G}} + 4 \leq 2\log^{1+\epsilon} |S|  + 4$ times
on the right-hand sides of $(S', G)$,
and if \Greedy finished the execution with smaller grammar
we have that no digram occurs more than once (without overlaps) on the right-hand sides of $(S', G)$.
Moreover, we can find at least $\frac{1}{2}(\rhs{S', G} - \nonterminals{G}-1)$ digrams
on the right-hand side of $(S', G)$, as we can pair symbols naively in each right-hand side;
we subtract $\nonterminals{G}$ factor because rules and starting string can have odd length.
Thus $S_G$ contains $\frac{1}{2}(|S_G| - \nonterminals{G}-1)$ digrams which occur at most
$2\log^{1+\epsilon} |S|  + 4$ times each.
Let $Y_{S_G}^P$ be a parsing of $S_G$ into phrases of length $1$ and $2$ where
the phrases of length $2$ correspond to the digrams mentioned above,
and let $L_P = \Lengths(Y_{S_G}^P)$. 
By applying Theorem~\ref{thm:main_estimation} with $k=0$ we get that:
\begin{align}
\label{eq:GreedyParsingEntropy}
|S_G|H_0(S_G) + |L_P|H_0(L_P) \geq |Y_{S_G}^P|H_0(Y_{S_G}^P) \enspace,
\end{align}
note that as $L_P \in \{1, 2\}^*$, it holds that
$|L_P|H_0(L_P) \leq |L_P|\log |\{1, 2\}| \leq |S_G|$.

On the other hand, as each phrase in $Y_{S_G}^P$ occurs at most 
$2\log^{1+\epsilon} |S| + 4$ times we can lower bound the entropy $|Y_S^P|H_0(Y_S^P)$, 
by lower bounding each summand corresponding to the phrase of length $2$ in
the definition of $H_0$
by $\log \frac{|Y_{S_G}^P|}{2\log^{1+\epsilon} |S|  + 4 }$, i.e\
\begin{align}
\label{eq:GreedyYSPandentropy}
|Y_{S_G}^P|H_0(Y_{S_G}^P)
&\geq \notag
\frac{1}{2}(|S_G| - \nonterminals{G}-1) \log\left(
\frac{|Y_{S_G}^P|}{2\log^{1+\epsilon} |S| + 4}\right)\\
&\geq \notag
\frac{1}{2}(|S_G| - \nonterminals{G}-1) \log\left(
\frac{ \frac{1}{2}(|S_G|-1) }{2\log^{1+\epsilon} |S| + 4}\right)\\
&\geq \notag
\frac{|S_G|}{2}\log\left( \frac{ |S_G|-1 }{4\log^{1+\epsilon} |S| + 8}\right)
- \frac{\nonterminals{G} + 1}{2}\log\left( \frac{ |S_G|-1 }{4\log^{1+\epsilon} |S| + 8}\right)\\
&\geq \notag
\frac{|S_G|}{2}\log\left( |S_G|-1\right)
- \frac{|S_G|}{2} \log (4\log^{1+\epsilon} |S| + 8)
- \gamma_1\\
\intertext{where $\gamma_1 = \frac{\nonterminals{G} + 1}{2}\log\left( \frac{ |S_G|-1 }{4\log^{1+\epsilon} |S| + 8}\right)$,
observe that as $\nonterminals{G} = \left\lceil \frac{|S|}{\log^{1 + \epsilon} |S|} \right\rceil$
and $|S_G|\leq |S|$ we have that $\gamma_1 = o(S)$ }
&\geq \notag
\frac{|S_G|}{2}\log\left( |S_G|-1\right)
- \gamma_2
- \gamma_1\\
\intertext{where $\gamma_2 = \frac{|S_G|}{2} \log (4\log^{1+\epsilon} |S| + 8)$,
observe that as $|S_G| = \Ocomp\left(\frac{|S|}{\log_\sigma |S|}\right)$
we have $\gamma_2 = o(|S|\log \sigma)$}
&\geq \notag
\frac{|S_G|}{2}\log |S_G|
- \frac{|S_G|}{2} \log \left( \frac{|S_G|}{|S_G|-1} \right)
- \gamma_2
- \gamma_1\\
&=
\frac{|S_G|}{2}\log |S_G| - \gamma \enspace ,
\end{align}
where $\gamma = \frac{|S_G|}{2} \log \left( \frac{|S_G|}{|S_G|-1} \right) + \gamma_2 + \gamma_1$.
Observe that as $\Ocomp\left(x\log\frac{x}{x-1}\right) = \Ocomp(1)$
we have that $\gamma = o(|S|\log \sigma)$.

Combining~\eqref{eq:GreedyParsingEntropy} and~\eqref{eq:GreedyYSPandentropy}
gives us lower bound on the entropy of $S_G$:
\begin{align}
\label{eq:greedy_sg_lower_bound}
|S_G|H_0(S_G) + |L_P|H_0(L_P) \geq \frac{|S_G|}{2}\log |S_G| - \gamma \notag\\
|S_G|H_0(S_G) \geq \frac{|S_G|}{2}\log |S_G| - \gamma' \enspace ,
\end{align}
where $\gamma'= \gamma + |L_P|H_0(L_P) = o(|S|\log \sigma)$.

We recall that $|S'_G| = \rhs{S'', G'} \leq \rhs{S', G} = |S_G|$
(as the size of the full grammar does not increase during the execution,
see Lemma~\ref{greedy_freq_does_not_increase}),
moreover in fully naive encoding
each symbol of the $(S'', G')$ is encoded
using at most $\lceil \log |S_G'| \rceil$ bits.
Thus, the fully naive encoding consumes at most:
\begin{align*}
|S'_G|\log|S'_G| + |S'_G|
&\leq |S_G|\log|S_G| + |S_G|\\	
&\leq 2|S_G|H_0(S_G) + 2\gamma' + |S_G| & \text{by~\eqref{eq:greedy_sg_lower_bound}} \\
&\leq 2|S|H_k(S) + o(|S|\log \sigma) \enspace, & \text{by~\eqref{eq:greedy_sg_upper_bound}}
\end{align*}
which finishes the proof for the case of fully naive encoding.

In the case of entropy coding, we make the same observation
as in the proof of Theorem~\ref{thm:RePairToTheEnd}:
we can look at the last iterations as if the algorithm had started
with string $S_G$ and would have finished with grammar $(S'', G'\setminus G)$,
note that concatenation of the right-hand sides of $(S'', G'\setminus G)$ gives $S'_G$.
As $(S'', G'\setminus G)$ meets the conditions of a weakly non-redundant grammar
and the entropy of $S_G$ is large~\eqref{eq:greedy_sg_lower_bound}
we can apply Lemma~\ref{lemma15replacemententropy} with $n=|S_G|$:
\begin{align*}
|S'_G|H_0(S'_G) &\leq \frac{3}{2}|S_G|H_0(S_G) + \Ocomp(|S_G| + \gamma')\\
&\leq \frac{3}{2} |S|H_k(S) + o(|S|\log \sigma) \enspace, &\text{by~\eqref{eq:greedy_sg_upper_bound}}
\end{align*}
thus by Lemma~\ref{lem:huffman_entropy_coder} the claim holds for the entropy coding.
\end{proof}

\section{Additional material for Section~\ref{sec:entropy_bound_for}}
\label{sec:appendix_c}
The idea of assigning values to phrases (see Definition~\ref{def:phrasecost})
was used before, for instance by Kosaraju and Manzini in 
estimations of entropy of \lzss and \lzso~\cite{KosarajuManzini99},
yet their definition depends on $k$ symbols
preceding the phrase.
This idea was adapted from
methods used to estimate the entropy of the source model~\cite{wyner1994sliding},
see also~\cite[Sec.~13.5]{ElementfofInformationTheory2006}.
A similar idea was used
in the construction of compressed text representations~\cite{GonzalezNStatistical}.
This work used a constructive argument:
it calculated the bit strings for each phrase
using arithmetic coding and context modeller.
Later 
it was observed that those two can be replaced with
Huffman coding~\cite{FerraginaV07SimpStat}.
Still, both of these representations were based on the assumption
that text is parsed into short (i.e\ $\alpha \log_\sigma n, \alpha < 1$) phrases of equal length
and used specific compression methods in the proof.

	Observe that the definition of phrase probability and parsing cost
	are also valid for $k=0$, as we assumed that
	$w[i\ldots j] =\epsilon$ when $i > j$, and $|S|_\epsilon = |S|$.

\begin{proof}[Proof Lemma~\ref{lem:cost_and_kentropy}]
	The lemma follows from the definition of $C_k(Y_S)$ and $H_k$:
	in $|S|H_k(S)$
	each letter $a$ occurring in context $v$ in $S$
	contributes $\log \prob(a |v)$
	 to the sum.
	Now observe that $C_k(Y_S)$ is almost the same as $|S|H_k(S)$,
	but the first $k$ letters of each phrase
	contribute $\log \sigma$
	instead of $\log \prob(a | v)$.
\end{proof}

\begin{proof}[Proof lemma~\ref{lem:mean_entr_cost}]
	Consider $l$ parsings $Y_S^0,\ldots,  Y_S^{l-1}$ of $S$
	where in $Y_S^i$ the first phrase has $i$ letters
	and the other phrases are of length $l$,
	except maybe the last phrase;
	to streamline the argument, we add a leading empty phrase to $Y_S^0$.
	Denote $Y_S^i = y_{i,0}, y_{i,1}, \ldots$.
	We estimate the sum $\sum_{i=0}^{l-1} C(Y_S^i)$.
	The costs of each of the first phrases is upper-bounded
	by $\log |S|$,
	as for any phrase $y$ the phrase cost $C(y)$ is at most $\log |S|$:
	%
	\begin{align*}
	C(y)
	&=
	- \log \prod _{j=1}^{|y|}  \prob(y[j] \ | \ y[1\ldots j])\\
	&\leq
	- \log \prod _{j=1}^{|y|} \frac{|S|_{y[1\ldots j]}}{|S|_{y[1\ldots j-1]}}\\
	\intertext{Note that by by Definition~\ref{def:k_order_entropy} some factors
	may be $\frac{|S|_{wa}}{|S|_w - 1}$, we bound them by $\frac{|S|_{wa}}{|S|_w}$,
	as $-\log \frac{|S|_{wa}}{|S|_w - 1} \leq -\log \frac{|S|_{wa}}{|S|_w}$  }
	&=
	-\log\left(\frac{|S|_y}{|S|_\epsilon}\right)\\
	&\leq
	\log(|S|/1) \enspace .
	\end{align*}
	Then $\sum_{i=0}^{l-1} C(Y_S^i)$ without the costs
	of all these phrases is:
	\begin{equation}
	\label{eq:sum_of_costs_sum_of_entropies}
	-\sum_{i=0}^{l-1}
	\sum_{p=1}^{|Y_S^i|}
	\sum_{j=1}^{|y_{i,p}|}\log
	\prob(y_{i,p}[j] \ | \ y_{i,p}[1\ldots j-1])
	\end{equation}
	and we claim that this is exactly $|S|\sum_{i=0}^{l-1}H_i(S)$.
	Together with the estimation of the cost of the first phrases
	this yields the claim,
	as
	\begin{equation*}
	\sum_{i=0}^{l-1} C(Y_S^i) \leq l \log |S| + |S|\sum_{i=0}^{l-1}H_i(S)
	\end{equation*}
	and so one of the $l$ parsings has a cost that is at most
	the right-hand side divided by $l$.
	
	To see that~\eqref{eq:sum_of_costs_sum_of_entropies}
	is indeed the sum of entropies
	observe for each position $m$ of the word
	we count $\log$ of the probability of this letter occurring in
	preceding $i= 0, 1, \ldots \min(l-1, m-1)$-letter context exactly once:
	this is clear for $m \geq l$, as the consecutive parsings are offsetted
	by one position;
	for $m < l$ observe that in $Y_0^S, Y_1^S, \ldots Y_{m-1}^S$
	the letter at position $m$ we count $\log$ of
	probability of this letter occurring in
	preceding $m-1, m-2, \ldots, 0$ letter context,
	while in $Y^S_m, Y^S_{m+1}, \ldots$ we include it in
	the first phrase, so it is not counted in~\eqref{eq:sum_of_costs_sum_of_entropies}.
\end{proof}

\begin{proof}[proof of Theorem~\ref{thmHC}]
	We prove only the first inequality,
	the proof of the second one is similar.
	
	Let $Y$ denote the set of different phrases of $Y_S$.
	For each $y\in Y$ define
	%
	\begin{equation*}
	p(y) = \frac{|L|_{|y|}}{|L|} \cdot \mathcal{P}(y) =  \frac{|L|_{|y|}}{|L|} \cdot 2^{-C(y)} \enspace .
	\end{equation*}
	We want to show that $\sum_{y \in Y} p(y) \leq 1$, that is, it satisfies
	the conditions of Lemma~\ref{theoremP}.
	
	First, we prove that  for each $l$ it holds that
	\begin{equation*}
	\sum_{y \in \Gamma^l} \mathcal{P}(y) \leq 1 \enspace .
	\end{equation*}
	This claim is similar to \cite[Lem.~A.1]{KosarajuManzini99},
	we prove it by induction on $l$.
	For $l=1$ we have
	\begin{align*}
	\sum_{s \in \Gamma} \mathcal P(s)
	&=
	\sum_{s \in \Gamma} 
	\prob(s \ | \ \epsilon)\\
	&= \sum_{s \in \Gamma} 
	\frac{|S|_s}{|S|}\\
	&=
	1 \enspace .
	\end{align*}
	For $l > 1$ we group elements in sum by their $l-1$ letter prefixes:
	\begin{align*}
	\sum_{y \in \Gamma^l} \mathcal{P}(y)
	&=
	\sum_{ y'\in \Gamma^{l-1}}
	\sum_{s \in \Gamma} \mathcal{P}(y's)\\
	&=
	\sum_{y'\in \Gamma^{l-1}}
	\sum_{s \in \Gamma} \mathcal{P}(y') \cdot 
	\prob(s \ | y') \\
	&=
	\sum_{y'\in \Gamma^{l-1}}
	\mathcal{P}(y') \cdot \sum_{s \in \Gamma}
	\prob(s \ | y') \\
	&\leq
	\sum_{y'\in \Gamma^{l-1}}
	\mathcal{P}(y') \cdot 1\\
	&\leq 1 \enspace ,
	\end{align*}
	where the last (and only) inequality follows from the induction hypothesis.
	Then	
	\begin{align*}
	\sum_{y \in Y} p(y)
	&=
	\sum_{y \in Y} \frac{|L|_{|y|}}{|Y_S|} \cdot \mathcal{P}(y)\\
	&=
	\sum_{l=1}^{|S|}
	\frac{|L|_l}{|Y_S|} \cdot \sum_{\substack{y: \: y \in Y\\|y| = l }}   \mathcal{P}(y)\\
	&\leq
	\sum_{l=1}^{|S|}
	\frac{|L|_l}{|Y_S|} \cdot 1\\
	&=1 \enspace .
	\end{align*}
	Thus $p$ satisfies the assumption of Lemma~\ref{theoremP}
	and we can apply it on $Y_S$ and $p$:
	\begin{align*}
	|Y_S| H_0(Y_S)
	&\leq
	- \sum_{i=1}^{|Y_S|} \log p(y_i) \\
	&=
	-\sum_{i=1}^{|Y_S|} \log \left(\frac{|L|_{|y_i|}}{|L|} \cdot \mathcal{P}(y_i)\right)\\
	&=
	- \sum_{i=1}^{|Y_S|} \log \frac{|L|_{|y_i|}}{|L|}
	- \sum_{i=1}^{|Y_S|} \log \mathcal{P}(y_i)\\
	&\leq
	|L|H_0(L)  +  \sum_{i=1}^{|Y_S|} C(y_i)\\
	&=
	|L|H_0(L)  +  C(Y_S) \enspace . \qedhere
	\end{align*}
\end{proof}

We now estimate the entropy of lengths $|L|H_0(L)$,
in particular in case of small parsings,
i.e.\ when $|Y_S| = o(|S|)$.
Those estimations include standard results on the entropy of lengths%
~\cite[Lemma~13.5.4]{ElementfofInformationTheory2006}
and some simple calculations.
(Note that a weaker estimation with stronger assumption was used implicitly in the work of Kosaraju and Manzini~\cite[Lemma A.3]{KosarajuManzini99}).

\begin{lemma}[Entropy of lengths]
	\label{lem:lengths_entropy}
	Let $S$ be a string, $Y_S$ its parsing and
	$L = \Lengths(Y_S)$. Then:
	\begin{equation*}
	|L|H_0(L) 
	\leq |L| \log \frac{|S|}{|L|} +  |L|(1 + \log \mathrm{e}) \enspace .
	\end{equation*}
	In particular, if $|Y_S| = \Ocomp\left( \frac{|S|}{\log_\sigma |S|} \right)$ then
	\begin{align*}
	|L| H_0(L) &= o(|S|\log \sigma) \enspace ,
	\intertext{and for any value of $|Y_S|$:}
	|L| H_0(L) &= \Ocomp(|S|) \enspace .
	\end{align*}
\end{lemma}

\begin{proof}
	Introduce random variable $U$ such that $Pr[U=l] = \frac{|L|_l}{|L|}$.
	Then $|L|H_0(L) = |L| H(U)$ and $E[U] = \frac{|S|}{|L|}$,
	where here $H$ is the entropy function for random variables
	and $E$ is the expected value.
	It is known~\cite[Lemma 13.5.4]{ElementfofInformationTheory2006} that:
	\begin{equation*}
	H(U) \leq  E[U+1] \log E[U+1] - E[U] \log E[U] \enspace .
	\end{equation*}
	Translating those results back to the setting
	of empirical entropy we obtain :
	\begin{align*}
	H_0(L)
	&\leq
	\left(1 + \frac{|S|}{|L|} \right)\log\left(1 + \frac{|S|}{|L|}\right) - \frac{|S|}{|L|} \log \frac{|S|}{|L|} \\
	&=
	\log \left( 1 +\frac{|S|}{|L|} \right) +
	\frac{|S|}{|L|} \log \left( \frac{|S|+|L|}{|L|} \cdot \frac{|L|}{|S|} \right) \\
	&\leq
	\log \left(2 \cdot \frac{|S|}{|L|}\right) +
	\frac{|S|}{|L|} \log \left( 1 + \frac{|L|}{|S|} \right)
	\\
	&\leq
	1 + \log \frac{|S|}{|L|} +
	\log\left( \left(1 + \frac{1}{|S|/|L|} \right)^{|S|/|L|}\right)\\
	&<
	1 + \log \frac{|S|}{|L|} + 	\log \mathrm e\enspace.
	\end{align*}
	Multiplying by $|L|$ yields the desired inequality.
	
	Moving to the second claim:
	assume $|Y_S| = \Ocomp\left( \frac{|S|}{\log_\sigma |S|} \right)$,
	thus $|L|=|Y_S| = \Ocomp\left( \frac{|S|}{\log_\sigma |S|} \right)$.
	Then:
	\begin{align*}
	|L|H_0(L) &< |L|\log \frac{\mathrm e |S|}{|L|} + |L|\\
			 &= \Ocomp\left( \frac{|S|}{\log_\sigma |S| } \log
			 \frac{|S|}{|S|/\log_\sigma |S|} \right) + |L|
			 & \text{as $x\log \frac{\mathrm e c}{x}$ is increasing in $x$ for $x \leq c$}
			  \\
			 &= \Ocomp\left( \frac{|S|\log \log_\sigma |S|}{\log_\sigma |S|} \right) + |L|\\
			 &= o(|S|\log \sigma)
	\end{align*}

	Concerning the last claim, when $|L|$ is arbitrary (but at most $|S|$),
	by definition $|L|(1+\mathrm{e}) = \Ocomp(|S|)$
	and the function $f(x) = x \log(c/x)$ is maximized for
	$x = c/\mathrm e$, for any $c$;
	for which it has value $c \frac{\log \mathrm e}{\mathrm e} \in \Ocomp(c)$,
	this is easily shown by computing the derivative.
	Thus we have:
	\begin{align*}
	|L|H_0(L) &< |L|\log \frac{|S|}{|L|} + |L|(1+\log \mathrm e) \\
			  &\leq |S|\frac{\log \mathrm e}{\mathrm e }  + |L|(1+\log \mathrm e)\\
			  &= \Ocomp(|S|) \enspace,
	\end{align*}
	so the lemma holds.
\end{proof}

\begin{proof}[Proofs of Theorem~\ref{thm:main_estimation} and Theorem~\ref{thm:mean_entropy}]
	The first inequality of Theorem~\ref{thm:main_estimation}
	follow from Theorem~\ref{thmHC} and Lemma~\ref{lem:cost_and_kentropy}:
	\begin{align*}
	|Y_S|H_0(S) &\leq C_k(Y_S) + |L|H_0(L)  & \text{by Theorem~\ref{thmHC}}\\
				&\leq |S|H_k(S) + |Y_S|k\log \sigma +  |L|H_0(L) \enspace .
				 &\text{by Lemma~\ref{lem:cost_and_kentropy}}
	\end{align*}
	The second inequality of Theorem~\ref{thm:main_estimation} comes directly form bounds in Lemma~\ref{lem:lengths_entropy}
	and the assumption that $k= o(\log_\sigma |S|)$.
	The bound on the encoding of string $S'$ follows
	directly from Lemma~\ref{lem:huffman_entropy_coder}.

	Similarly, Theorem~\ref{thm:mean_entropy} follows from
	Theorem~\ref{thmHC} and Lemma~\ref{lem:mean_entr_cost}:
	\begin{align*}
		|Y_S|H_0(Y_S) &\leq C(Y_S) + |L|H_0(L) & \text{by Theorem~\ref{thmHC}} \\
					  &\leq \frac{|S|}{l}\sum_{i=0}^{l-1} H_i(S) + \log |S| +  |L|H_0(L)
					  	& \text{by Lemma~\ref{lem:mean_entr_cost}}
	\end{align*}
	In this case, we can bound $|L|H_0(L)$ by $\Ocomp(\log |S|)$
	as parsing from Lemma~\ref{lem:mean_entr_cost}
	consist of phrases of the same length, except for the first and the last one,
	which can be shorter:
	\begin{equation*}
	|L|H_0(L) \leq |L|\log \frac{|L|}{|L|-2} + 2\log |L| = \Ocomp(\log |S|) \qedhere
	\end{equation*}
\end{proof}

\end{document}